%% file: datashare-mspsi.tex
\documentclass[letterpaper,twocolumn,10pt]{article}
\usepackage{usenix2019_v3}

\usepackage{url}
\usepackage{tikz}
\usepackage{amsmath}
\usepackage{amsthm}
\usepackage{filecontents}

\usepackage{caption} 
\usepackage{booktabs}

\usepackage[scaled=0.88]{helvet}  

\usepackage[T1]{fontenc}
\usepackage{float}
\usepackage{amssymb}
\usepackage{listings}
\usepackage{graphicx}
\usepackage{subfigure}
\usepackage{adjustbox}
\usepackage{xspace}
\usepackage{cite}
\usepackage{tabularx}
\usepackage{enumitem}
\usepackage{threeparttable}
\usepackage{enumitem}
\usepackage{mathrsfs}
\usepackage[ruled]{algorithm}
\usepackage[noend]{algpseudocode}
\usepackage{mathtools}
\usepackage[n,operators,advantage,sets,adversary,landau,probability,notions,logic,ff,mm,primitives,events,complexity,asymptotics,keys]{cryptocode}
\usepackage{comment}

\excludecomment{confversion}\includecomment{fullversion}

\begin{document}

\input{resource/macros.tex}

\title{\Large \bf \nameN\\ A Decentralized Privacy-Preserving Search Engine for Investigative Journalists}

\date{}

\author{
    {Kasra EdalatNejad} \\ SPRING Lab, EPFL
    \and 
    {Wouter Lueks} \\ SPRING Lab, EPFL
    \and 
    {Julien Pierre Martin} \\ Independent
    \and 
    {Soline Led\'esert} \\ ICIJ
    \and 
    {Anne L'H\^ote} \\ ICIJ
    \and 
    {Bruno Thomas} \\ ICIJ
    \and 
    {Laurent Girod} \\ SPRING Lab, EPFL
    \and 
    {Carmela Troncoso} \\ SPRING Lab, EPFL
    \and 
} 

\maketitle
\input{parts/s0-abstract.tex}
\input{parts/s1-introduction.tex}

\input{parts/s2-system.tex}
\input{parts/s3-mspsi.tex}

\input{parts/s4-messaging.tex}

\input{parts/s5-datashare.tex}
\input{parts/s6-evaluation.tex}

\input{parts/s7-related-work.tex}
\input{parts/s8-conclusions.tex}

\bibliographystyle{plain}
\bibliography{datashare-mspsi}

\appendix

  \input{parts/s11-proof.tex}
  \input{parts/s10-appendix.tex}


\end{document}

%% file: resource/macros.tex
\setlength\marginparwidth{16mm}
\setlength\marginparsep{1mm}
\setlength\marginparpush{2mm}
\newcommand{\commentmarker}[1]{\colorbox{red}{\textcolor{white}{#1}}}
\newcounter{CommentCounter}
\newcommand{\commentUp}[3][]{%
  \stepcounter{CommentCounter}%
  {\scriptsize\commentmarker{\theCommentCounter}}%
  \marginpar{%
    \vspace{-#2}
    \tiny\raggedright
    \commentmarker{\theCommentCounter}
    \ifstrempty{#1}{}{\textcolor{red}{#1: }}%
    {\scriptsize #3}%
  }%
}
\newcommand{\internalcomment}[2][]{\commentUp[#1]{0pt}{#2}}

\newcommand{\done}[1]{\ifdefined\hidetodo \else{\color{lime}{ done: #1} }\fi}
\newcommand{\todo}[1]{\ifdefined\hidetodo \else{\color{red}{ todo: #1} }\fi}
\newcommand{\TODO}[1]{\ifdefined\hidetodo \else{\color{red}{\textbf{ TODO: #1} }}\fi}
\newcommand{\note}[1]{\ifdefined\hidetodo \else{\color{teal}{ (Note: #1)} }\fi}
\newcommand{\depricate}[1]{\ifdefined\showdepricate \greyout{ #1 } \else \fi}
\newcommand{\greyout}[1]{{\color{gray}{ {\footnotesize #1 }} }}
\newcommand{\needcite}{\ifdefined\hidetodo \else{\color{red}{[?]}}\fi}
\newcommand{\needref}{\ifdefined\hidetodo \else{\color{red}{ref?}}\fi}
\newcommand{\authsec}{\note{auth section}}
\newcommand{\EDCcite}{De Cristofaro et al. \cite{CristofaroGT12}\xspace}
\newcommand{\name}{\textsc{Datashare}\xspace}
\newcommand{\nameN}{\textsc{DatashareNetwork}\xspace}
\newcommand{\DataShare}{\name}
\newcommand{\ICIJ}{ICIJ\xspace}
\newcommand{\ICIJFull}{International Consortium of Investigative Journalists\xspace}
\newcommand{\MIS}{CS\xspace}
\newcommand{\DsIR}{Datashare Multilingual Information Extraction Platform\xspace}
\newcommand{\paranodot}[1]{\vspace{1mm}\noindent\textbf{#1}}
\newcommand{\para}[1]{\vspace{1mm}\noindent\textbf{#1}.}

\newcommand{\confversioncmd}[1]{}
\newcommand{\fullversioncmd}[1]{}
\begin{confversion}
  \renewcommand{\confversioncmd}[1]{#1}
\end{confversion}
\begin{fullversion}
  \renewcommand{\fullversioncmd}[1]{#1}
\end{fullversion}

\renewcommand{\psi}{PSI\xspace}
\newcommand{\psica}{PSI-CA\xspace}
\newcommand{\cspsi}{C-PSI\xspace}
\newcommand{\cspsica}{C-PSI-CA\xspace}
\newcommand{\mspsi}{MS-PSI\xspace}
\newcommand{\mspsica}{MS-PSI\xspace}
\newcommand{\mspsicalim}{MS-PSI-lim\xspace}
\newcommand{\numdocsearch}{\#doc\xspace}

\newcommand{\yes}{\checkmark}
\newcommand{\no}{$\times$}
\newcommand{\soso}{$\sim$}
\newcommand{\order}{\mathcal{O}}
\newcommand{\ith}{$i^{\text{\tiny{th}}}$}

\newcommand{\gen}{g}
\newcommand{\G}{\mathbb{G}} 
\newcommand{\grouporder}{p}
\newcommand{\Zp}{\mathbb{Z}_{\grouporder}}
\newcommand{\Zn}{\mathbb{Z}_{N}}
\newcommand{\randin}{\gets^{\$}}
\newcommand{\permutation}{\Pi}
\newcommand{\binanswer}{0/1}
 
\renewcommand{\hash}{H}
\newcommand{\GHash}{\hat{H}}
\newcommand{\cuckoo}{\textvar{CF}}
\renewcommand{\secpar}{\ell}
 
\newcommand{\code}[1]{\text{\textbf{#1}}}
\newcommand{\codify}[1]{ \text{\textsf{#1}} }

\newcommand{\NatNumUpTo}[1]{[#1]}
\newcommand{\uniformdist}[1]{[0, #1]}

\newcommand{\kand}{ \code{ \& } }
\newcommand{\kforeach}[2]{{\code{For } #1 \code{ in } #2 :}}
\newcommand{\foriter}[2]{#1=1,..,{#2}}
\newcommand{\setdelim}{\;|\;}
\renewcommand{\concat}{\;||\;}
\newcommand{\makeset}[2]{\set{#1{1}, #1{2}, .., #1{#2}} }
\newcommand{\makesetsub}[2]{\set{#1_{1}, .., #1_{#2}} }
\newcommand{\makelist}[2]{[#1{1}, #1{2}, .., #1{#2}]}

\newcommand{\expCost}[1]{\ensuremath{#1 \tau_e}}
\newcommand{\hashCost}[1]{\ensuremath{#1 \tau_H}}
\newcommand{\exphashCost}[1]{\ensuremath{#1\tau_{H+e}} }
\newcommand{\millisec}{ms}
\newcommand{\second}{s}
\newcommand{\minute}{Min}

\newcommand{\serverkey}{s}
\newcommand{\clientkey}{c}

\newcommand{\clientsize}{m}
\newcommand{\serversize}{N}
\newcommand{\setsize}[1]{{n_{#1}}}
\newcommand{\totalsize}{S}
\newcommand{\kwlim}{\textvar{lim}}
\newcommand{\uniqueNum}{u}

\newcommand{\y}{y}
\newcommand{\x}{x}
\newcommand{\xp}{\tilde{x}}
\newcommand{\xpp}{\hat{x}}
\newcommand{\collection}{ \mathcal{Y}  }
\newcommand{\yset}{Y}
\newcommand{\yseti}[1]{\yset_{#1}}
\newcommand{\pretagset}[1]{\tau^{(#1)}}
\newcommand{\compPretag}{\pretag=\x^{\serverkey}}
\newcommand{\compFinaltag}{\finaltag_i=\hash(i \concat \pretag)}
\newcommand{\cardinality}{ \textrm{Out}  }
\newcommand{\tagCollection}{\textvar{TC}}
\newcommand{\clientTags}{\textvar{T}}
\newcommand{\clientPreTag}{\tau}

\newcommand{\X}{X}
\newcommand{\XP}{\tilde{X}}
\newcommand{\XPP}{\hat{X}}
\newcommand{\TAG}{\textrm{Tags}}

\newcommand{\XDef}{\X       = \makesetsub{\x   }{\clientsize}}
\newcommand{\XYintersection}{\X \cap \yset}
\newcommand{\YSETDef}{ \yset_i = \set{\y_{i,1}, .., \y_{i,\setsize{i}} } }
\newcommand{\CollectionDef}{ \collection = \makesetsub{\yset}{\serversize} }

\newcommand{\n}{n}
\newcommand{\kw}[1]{a_{#1}}
\newcommand{\kwi}{\kw{i}}
\newcommand{\kwx}{\kw{x}}
\newcommand{\kwlist}{\{\kw{1}, .., \kw{\n}\}}
\newcommand{\kwset}{S}
\newcommand{\partialset}{P}

\newcommand{\gssetup}{\textvar{AC.setup}}
\newcommand{\gsjoin}{\textvar{AC.depricatedddddd}}
\newcommand{\gsissue}{\textvar{AC.obtain}}
\newcommand{\gsverify}{\textvar{AC.verify}}
\newcommand{\gsreveal}{\textvar{AC.reveal}}
\newcommand{\gsmsk}{\textvar{msk}}
\newcommand{\gsmpk}{\textvar{mpk}}
\newcommand{\gsusersk}{\textvar{sk}}
\newcommand{\gsattr}{\textvar{attrs}}
\newcommand{\gsratelim}{\textvar{rate\_lim}}

\newcommand{\gssign}{\textvar{Sign}}
\newcommand{\gsverifysig}{\textvar{Verify}}

\newcommand{\gsacc}{\textvar{Acc}}
\newcommand{\gsrej}{\textvar{Rej}}
\newcommand{\signature}{\sigma}
\newcommand{\cred}{C}

\newcommand{\cfinputset}{S}
\newcommand{\cfelem}{x}
\newcommand{\cfparams}{\textvar{params}}
\newcommand{\cfcompress}{\textvar{CF.compress}}
\newcommand{\cfcheck}{\textvar{CF.membership}}
\newcommand{\cfintersection}{\textvar{CF.intersection}}

\newcommand{\sampleclientsize}{10}
\newcommand{\sampleserversize}{100}
\newcommand{\samplesetsize}{1000}
\newcommand{\sampleuniqueratio}{10}
\newcommand{\sampleuniquesize}{10,000}
\newcommand{\sampletotalsize}{100,000}

\newcommand{\field}[2]{{#2}_{#1}}
\newcommand{\object}[1]{}
\newcommand{\jRecord}{\textvar{Rec}}
\newcommand{\jour}{J}
\newcommand{\owner}{O}
\newcommand{\doc}{d}
\newcommand{\journalists}{\mathcal{J}}
\newcommand{\jourField}[1]{\field{\jour}{{#1}} }
\newcommand{\ownerField}[1]{\field{\owner}{{#1}} }
\newcommand{\ownerNymField}[1]{\field{\ownerField{\pid}}{{#1}} }
\newcommand{\querierField}[1]{\field{q}{{#1}} }
\newcommand{\queryField}[1]{\field{\query}{{#1}} }
\newcommand{\querySignature}{\signature_q}
\newcommand{\contactSignature}{\field{C}{\signature}}
\newcommand{\query}{Q}
\newcommand{\querier}{\mathcal{Q}}
\newcommand{\response}{R}
\newcommand{\pid}{\textvar{nym}}
\newcommand{\publicMailbox}{\textvar{BB}}
\newcommand{\DSPH}{\textvar{MS}\xspace}
\newcommand{\DSpigeonhole}{communication server\xspace}
\newcommand{\docMatch}{t}
\newcommand{\emailmb}{\textvar{email}}

\newcommand{\jourPK}{\textsf{pk}}
\newcommand{\jourSK}{\textsf{sk}}
\newcommand{\pkEnc}{\textvar{Enc}}
\newcommand{\pkDec}{\textvar{Dec}}

\newcommand{\symKey}{\textrm{key}}
\newcommand{\symDec}{\textrm{Dec}}
\newcommand{\symEnc}{\textrm{Enc}}
\newcommand{\anonCred}{\sigma}
\newcommand{\thread}{\textrm{Th}}
\newcommand{\createThread}{\textrm{createThread}}
\newcommand{\jourGS}{\textrm{JGS}}

\theoremstyle{definition}
\newtheorem{definition}{Definition}
\newtheorem{protocol}{Protocol}
\newtheorem{lemma}{Lemma}
\newtheorem{theorem}{Theorem}
\newcommand{\parab}[1]{\vspace{1mm}\noindent\textbf{#1}}          
\newcommand{\parait}[1]{\noindent\textit{#1}}                     
\newcommand{\paraNormal}[1]{\vspace{1mm}\noindent\normalfont{#1}} 
\newcommand{\textvar}[1]{\textrm{\textsf{#1}}} 
\newcommand{\tikzlongarrow}[2]{
\begin{tikzpicture}[]
  \draw[#1](0,0) -- node[above=-0.5ex]{\ensuremath{#2}} (1.5,0);
  \node[draw=none] (bottom) at (0,-0.5ex) {};
  \node[draw=none] (top) at (0,1ex) {};
  \node[draw=none] (lowerleft) at (bottom-|current bounding box.west) {};
  \node[draw=none] (topright) at (top-|current bounding box.east) {};
  \pgfresetboundingbox
  \draw[draw=none,use as bounding box] (lowerleft) rectangle (topright);
\end{tikzpicture}
}
\newcommand{\diagramsend}[1]{\tikzlongarrow{->}{#1}}
\newcommand{\diagramrecv}[1]{\tikzlongarrow{<-}{#1}}

%% file: parts/s0-abstract.tex
\begin{abstract}
  Investigative journalists collect large numbers of digital documents during
  their investigations. These documents can greatly benefit other journalists'
  work. However, many of these documents contain sensitive information. Hence,
  possessing such documents can endanger reporters, their stories, and their
  sources. Consequently, many documents are used only for single, local,
  investigations.
  We present \nameN, a decentralized and privacy-preserving search system
  that enables journalists worldwide to find documents via a dedicated network
  of peers. \nameN combines well-known anonymous authentication mechanisms and
  anonymous communication primitives, a novel asynchronous messaging system, and
  a novel multi-set private set intersection protocol (MS-PSI) into a
  \emph{decentralized peer-to-peer private document search engine}. 
  We prove that \nameN is secure; and show using a prototype implementation
  that it scales to thousands of users and millions of documents.

\end{abstract}

%% file: parts/s1-introduction.tex
\section{Introduction}
Investigative journalists research topics such as corruption, crime, and
corporate misbehavior. Two well-known examples of investigative projects are the
Panama Papers that resulted in several politicians' resignations and sovereign
states recovering hundreds of millions of dollars hidden in offshore
accounts~\cite{Panama}, and the Boston Globe investigation on child abuse that
resulted in a global crisis for the Catholic Church~\cite{ChurchAbuses}.
Investigative journalists' investigations are essential for a healthy
democracy~\cite{Carson12}. They provide the public with information kept secret
by governments and corporations. Thus, effectively holding these institutions
accountable to society at large.

In order to obtain significant, fact-checked, and impactful results, journalists
require large amounts of documents. In a globalized world, local issues are
increasingly connected to global phenomena. Hence, journalists' collections can
be relevant for other colleagues working on related investigations. However,
documents often contain sensitive and/or confidential information and possessing
them puts journalists and their sources increasingly at risk of identification,
prosecution, and persecution~\cite{McGregorRC16,McGregorCHR15}. As a result
journalists go to great lengths to protect both their documents and their
interactions with other journalists~\cite{McGregorWACR17}. With these risks in
mind, the \ICIJFull (\ICIJ) approached us with this question: \emph{Can a global
community of journalists search each other's documents while minimizing the risk
for them and their sources?} 

Building a practical system that addresses this question entails solving five key challenges:

\noindent1) \emph{Avoid centralizing information}. A party with access to all
the documents and journalists' interaction would become a very tempting target
for attacks by hackers or national agencies, and for legal cases and subpoenas
by governments. 

\noindent2) \emph{Avoid reliance on powerful infrastructure}. 
Although \ICIJ has journalists worldwide, it does not have highly available
servers in different jurisdictions.

\noindent3) \emph{Deal with asynchrony and heterogeneity}. Journalists are
spread around the world. There is no guarantee that they are online at the same
time, or that they have the same resources.

\noindent4) \emph{Practical on commodity hardware}. Journalists must be able to
search documents and communicate with other journalists without this affecting
their day-to-day work. The system must be efficient both computationally and
in communication costs.
 
\noindent5) \emph{Enable data sovereignty}. Journalists are willing to share but
not unconditionally. They should be able to make informed decisions on revealing
documents, on a case-by-case basis.

The first four requirements preclude the use of existing advanced
privacy-preserving search technologies, whereas the fifth requirement precludes
the use of automatic and rule-based document retrieval.
More concretely, the first requirement prevents the use of central databases and
private information retrieval (PIR)~\cite{BeimelI01,Goldberg07,KushilevitzO97}
between journalists, as standard PIR requires a central list of all searchable
(potentially sensitive) keywords.
The second requirement rules out multi-party computation (MPC) between
distributed servers~\cite{HuangEK12,Pinkas0WW18,Pinkas0SZ15}.

The third and fourth requirement exclude technologies that require
many round trips or high bandwidth between journalists such as custom private
set intersection~\cite{CristofaroT10,FalkNO18,HuangEK12,KissL0AP17,Pinkas0Z14},
keyword-based PIR~\cite{AngelS16,ChorGN97}, and generic MPC
protocols~\cite{HuangEK12,Pinkas0WW18,Pinkas0SZ15,WangRK17}, as well as
the use of privacy-preserving communication systems that require all users to be
online~\cite{LazarGZ18,HooffLZZ15}.

We introduce \nameN, a decentralized document search engine for journalists to
be integrated within \ICIJ's open source tool for organizing information called
Datashare\cite{DatasharePlatform}. \nameN addresses the challenges as follows.
First, journalists keep their collections in their computers. Thus, if a
journalist is hacked, coerced, or corrupted, only her collection is compromised.
Second, we introduce a new multi-set private set intersection (MS-PSI) protocol
that enables asynchronous search and multiplexes queries to reduce computation
and communication costs. Third, we combine existing privacy-preserving
technologies~\cite{DingledineMS04,AngelCLS18} to build a pigeonhole-like
communication mechanism that enables journalists to anonymously converse with
each other in an unobservable manner. These components ensure that even if an
adversary gains the ability to search others' documents, she cannot extract all
documents nor all users in the system. In the rest of the document, for
simplicity, we refer to \nameN as \name.

Our contributions are as follows:

\paranodot{\checkmark} We elicit the security and privacy requirements of a
  document search system for investigative journalists.

\paranodot{\checkmark} We introduce MS-PSI, a private set intersection protocol
  to efficiently search in multiple databases without incurring extra leakage
  with respect to traditional PSI with pre-computation. 

\paranodot{\checkmark} We propose an asynchronous messaging system that enables
journalists to search and converse in a privacy-preserving way.

\paranodot{\checkmark} We design \name, a secure and privacy-preserving
  decentralized document search system that protects from malicious users and
  third parties the identity of its users, the content of the queries and, to a
  large extent, the journalists' collections themselves. We show that \name
  provides the privacy properties required by journalists, and that the system
  can easily scale to more than 1000 participants, even if their document
  collections have more than 1000 documents.

%% file: parts/s2-system.tex
\section{Towards Building \name} 
\label{sec:system}

We build \name at the request of the \ICIJFull, \ICIJ. When unambiguous from the
context, we refer to \ICIJ simply as the organization.

\subsection{Requirements Gathering}
\label{sec:system:requirements}
In order to understand the needs of investigative journalists, \ICIJ ran a
survey among 70 of their members and provided us with aggregate statistics,
reported below. We used the survey results as starting point for the system's
requirements, and we refined these requirements in weekly meetings held for more
than one year with the members of \ICIJ's Data \& Research Unit who are in
charge of the development and deployment of the local tool
Datashare\cite{DatasharePlatform}.
 
\para{User Base} \ICIJ consists of roughly 250 permanent journalist members in
84 countries. These members occasionally collaborate with external reporting
partners. The maximum number of reporters working simultaneously on an
investigation has reached 400. The organization estimates that each member is
willing to make approximately one thousand of their documents available for
searching. To accommodate growth, we consider that \name needs to \emph{scale to
(at least) 1000 users, and (at least) 1 million documents}.

Journalists work and live all over the globe, ranging from Sydney to San
Francisco, including Nairobi and Kathmandu; this results in large timezone
differences. Around 38\% of the journalists have a computer permanently
connected to the Internet, and another 53\% of them are connected during work
hours: eight hours a day, five days a week. The rest are connected only during a
few hours per day. As it is unlikely that journalists are online at the same
time, \emph{the search system needs to enable asynchronous requests and
responses}. Furthermore, many journalists live in regions with low-quality
networks: only half of the journalists report having a fast connection. Thus,
\emph{\name cannot require high bandwidth}.

\para{Waiting Time} As the system must be asynchronous, the survey asked
journalists how much they are willing to wait to obtain a the result of a query.
21\% of the surveyees are willing to wait for hours, whereas another 56\% can
wait for one or more days. Hence, \emph{\name does not need to enable real-time
search}. Yet, given the delivery times of up to 24 hours, to keep search latency
within a few days, \name must \emph{use protocols that can operate with just one
communication round}. Therefore, we discard multi-round techniques such as
multi-party computation~\cite{HuangEK12,Pinkas0WW18,Pinkas0SZ15,WangRK17}.

\para{Queries Nature} The queries made by journalists are in a vast majority
formed by \emph{keywords} called \emph{named entities}: names of organizations,
people, or locations of interest. Therefore, journalists do not require a very
expressive querying language: \name must \emph{support queries made of
conjunctions of keywords}. Journalists are interested in a small set of these
entities at a time: only those related to their current project. Consequently,
\emph{queries are not expected to include more than 10 terms at a time}, and
\emph{journalists are not expected to issue a large number of queries in
parallel}. 

During the design phase, we also learned that as most terms of interest are
investigation-specific (e.g., XKeyScore in the Snowden leaks, or Mossack Fonseca
in the Panama Papers), \emph{a pre-defined list of terms cannot cover all
potentially relevant keywords for journalists}. Therefore, techniques based on
fixed lists such as private information retrieval
(PIR)~\cite{BeimelI01,Goldberg07,KushilevitzO97} are not suitable for building
\name.

\para{Security and Privacy} Regarding security and privacy concerns, journalists
identify four types of principals: the journalists themselves, their sources,
the people mentioned in the documents, and the ICIJ. They identify three assets:
the named entities in documents, the documents themselves, and the conversations
they have during an investigation. The disclosure of named entities could leak
information about the investigation, or could harm the cited entities (which
could in turn could trigger a lawsuit). Whole documents are considered the most
sensitive as they provide context for the named entities. Finally, the
disclosure of the content or existence of conversations could endanger the
journalists involved, their sources, the organization, and the whole
investigation.

Journalists mostly worry about third party adversaries such as corporations,
governments (intelligence agencies), and organized crime. Sources and other
journalists are in general considered non-adversarial. Similarly, journalists
\emph{trust ICIJ to be an authority for membership and to run their
infrastructure}. However, to prevent coercion and external pressures, ICIJ does
\emph{not} want to be trusted for privacy.

The main requirement for \name is \emph{to protect the confidentiality of assets
from third parties that are not in the system}. This implies that \name cannot
require journalists to send their data to third parties for analysis, storage,
indexing, or search. Journalists are concerned about only subsets of these
adversaries at a time. Therefore, \name \emph{does not need to defend against
global adversaries.}

Journalists initially did not consider their colleagues as adversaries. However,
after a threat analysis, we concluded that there is a non-negligible risk that
powerful adversaries can bribe or compromise honest journalists, in particular
when those journalists live in jurisdictions with less protection for civil
rights. Therefore, we require that \name \emph{must minimize the amount of
information that journalists, or ICIJ, learn about others: searched keywords,
collections, and conversations}. More concretely, we require that \emph{searches
be anonymous and that the searched terms be kept confidential, with respect to
both journalists and the organization.} This way neither journalists nor the
organization become a profitable target for adversaries.

With respect to conversations, 64\% of the surveyees report that they would
prefer to remain anonymous in some cases. Furthermore, 60\% of the respondents
declare that they prefer to have a screening conversation before deciding to
share documents. This means that \emph{search and sharing features need to be
separated to enable screening}. \name must \emph{provide anonymous means for
journalists to discuss document sharing to ensure safety}. We expect
conversations within \name to be short, as their only goal is to agree on
whether to proceed with sharing. After journalists agree, we assume they will
switch to an alternative secure communication channel and \name does not need to
support document retrieval.

\subsection{Sketching \name}
\name is run by \ICIJ. Access to the system is exclusive to \ICIJ members and
authorized collaborators. Journalists trust \ICIJ to act as a token issuer and
only give tokens to authorized journalists. To enable journalists to remain
anonymous, tokens are implemented using blind signatures. Journalists use these
tokens demonstrate membership without revealing their identities.

\name provides the following infrastructure to facilitate asynchronous
communication between journalists: a \emph{bulletin board} that journalists use
to broadcast information, and a \emph{pigeonhole} for one-to-one communication.
All communications between journalists and the infrastructure (pigeonhole or
bulletin board) are end-to-end encrypted (i.e., from journalist to journalist)
and anonymous. Hence, the infrastructure needs to be trusted for availability,
but not to protect the privacy of the journalists and their documents.

Each authorized journalist in \name owns a corpus of documents that they make
available for search. Journalists can take two roles: (i) \emph{querier}, to
search for documents of interest, and (ii) \emph{document owner}, to have their
corpus searched. Journalists first search for matching documents then
(anonymously) converse with the corresponding document owners to request the
document.

\begin{figure}[t]
    \includegraphics[width=\linewidth,trim=3mm 33mm 2mm 0mm,clip]{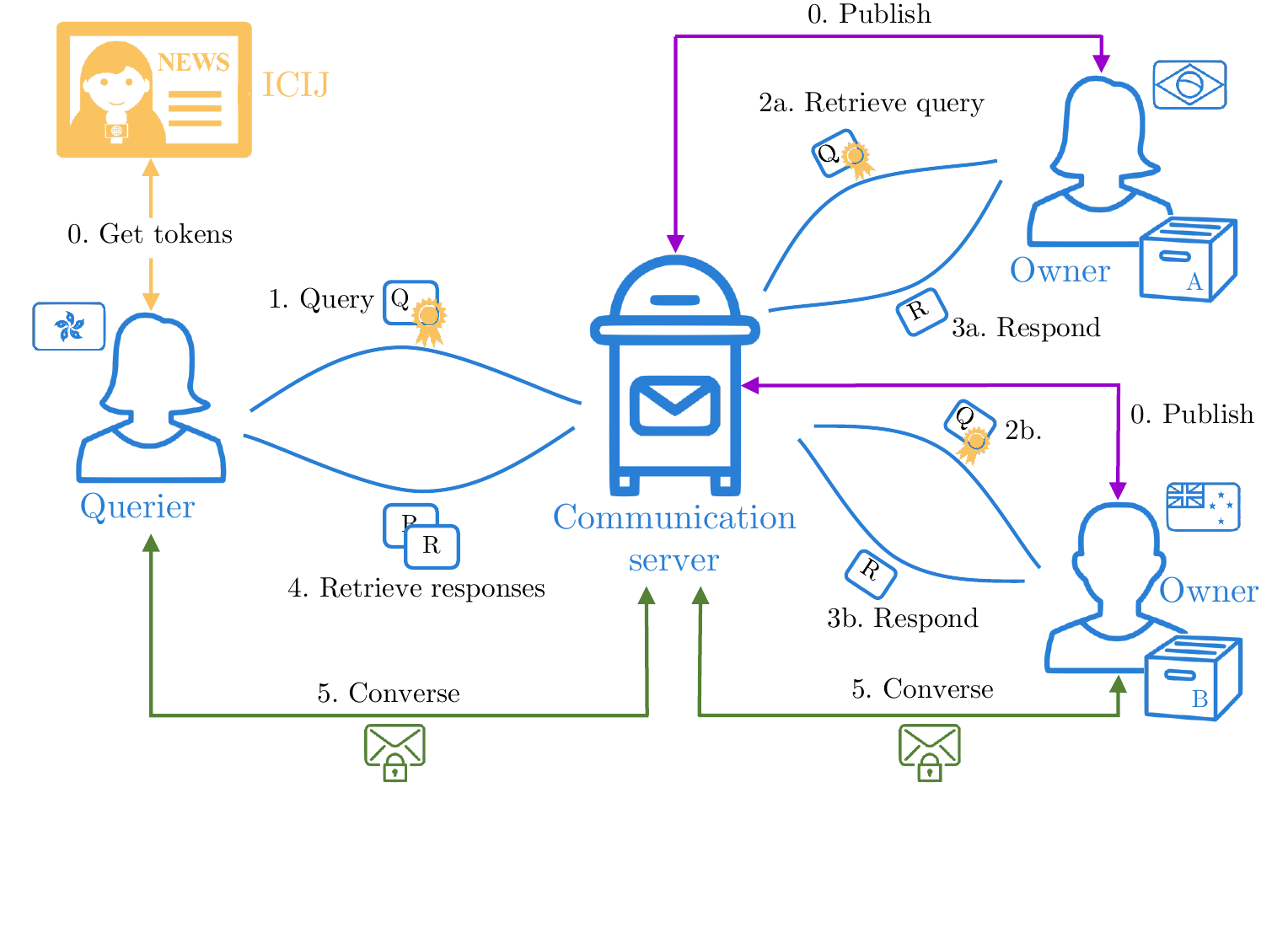}
    \caption{\name architecture overview.}
    \label{fig:arch}
\end{figure}

Figure~\ref{fig:arch} sketches \name's architecture. First, journalists upload
privacy-preserving representations of their collections and contact information
to the bulletin board. To issue a query, journalists construct a
privacy-preserving representation of their keywords and broadcast it together
with an authorization token through the bulletin board. Owners periodically
retrieve new queries from the bulletin board. If the authorization is valid,
they send a response to the querier using the pigeonhole. The querier uses this
response to identify matches with the documents in the owner's collection.

When journalists find a \emph{match} in a collection, i.e., a document that
contains all the keywords in the query, they can start a conversation with the
document owner to request sharing. Document owners append a public contact key
to their collection to enable queriers to carry out this conversation in an
anonymous way via the pigeonhole.

\para{Instantiation} \name uses four main privacy-preserving building blocks: a
multi-collection search mechanism, a messaging system, an anonymous
communication channel, and an authorization mechanism.

We implement the privacy-preserving search mechanism by using a novel primitive
that we call multi-set private set intersection (MS-PSI) described in
Section~\ref{sec:multi_set_psi}. We design a privacy-preserving messaging system
in Section~\ref{sec:messaging}; it provides both the bulletin board and
pigeonhole functionality. We rely on the Tor~\cite{DingledineMS04} network as
anonymous communication channel, and we use blind signatures to implement
privacy-preserving authorization (see Section~\ref{sub:preliminaries}). In
Section~\ref{sub:designing_datashare}, we explain how \name combines these
building blocks.

%% file: parts/s3-mspsi.tex
\section{Multi-set PSI} 
\label{sec:multi_set_psi}
Private set intersection (\psi) protocols enable two parties holding sets $\X$
and $\yset$ to compute the intersection $\XYintersection$, without revealing
information about the individual elements in the sets. In this section, we
introduce a multi-set private set intersection (\mspsi) protocol that
simultaneously computes intersections of set $\X$  with $\serversize$ sets
$\{\yseti{1},\ldots, \yseti{\serversize}\}$ at the server. In
Section~\ref{sec:related_work}, we review existing \psi variants.

\para{Notation} (See Table~\ref{table:notation}) We use a cyclic group $\G$ of
prime order $\grouporder$ generated by $\gen$. We write $x \randin X$ to denote
that $x$ is drawn uniformly at random from the set $X$. Let $\secpar$ be a
security parameter. We define two hash functions $\hash: \bin^* \rightarrow
\bin^{\secpar}$ and $\GHash: \bin^* \rightarrow \G$. Finally, we write
$\NatNumUpTo{n}$ to denote the set $\{1, \ldots, n\}$.

\input{resource/notation-table.tex}

\para{Related \psi Schemes}
\label{sub:naive_psi}
We build on the single-set \psi protocol by De Cristofaro et
al.~\cite{CristofaroGT12}, see Figure~\ref{fig:psi-proto-emiliano}. In this
protocol the client blinds her elements $\x_i \in \G$ as
$\xp_i=\x_i^{\clientkey}$ using a blinding factor $\clientkey$ before sending
them to the server. The server applies its own secret to the blinded elements,
$\xpp_i=\xp_i^{\serverkey}$, and sends them back to the client in the same
order, together with a tag collection of her own blinded elements:
$\tagCollection = \set{\hash(\y^{\serverkey}) \;|\; \y \in \yset}$. The client
unblinds her elements, obtaining a list of $\x_i^{\serverkey}$s. Then, the
client computes a tag $\hash(\x_i^{\serverkey})$ for each of them and compares
it to the server's tags $\tagCollection$ to find matching elements.

\input{resource/protocol-psi-emiliano.tex}

To increase efficiency when the server set is large, client-server \psi (\cspsi)
schemes in the literature~\cite{FalkNO18,KissL0AP17,Resende2018} introduce
optimizations to avoid that the server has to compute and send a large fresh set
of tags every execution. Instead, the server \emph{precomputes} the tag
collection with a long-term secret key $\serverkey$ and sends it to the client
once. In subsequent \emph{online} phases, the server answers clients' queries by
using the long-term key $\serverkey$. This significantly improves the
communication and computation cost, as the server does not compute or send the
tag collection every time.

\para{A New Multi-set PSI Protocol} 
Our \emph{multi-set} private set intersection protocol (\mspsi) intersects a
client set $\XDef \subset \bin^* $ with $\serversize$ sets $\YSETDef \subset
\bin^*$ at the server to obtain the intersections $\X \cap \yseti{i}$. Our
protocol computes all intersections \emph{simultaneously}, lowering the
computation and communication cost with respect to running $\serversize$
parallel \psi protocols. In \name, $X$ contains the query (a conjunction of
search keywords) and $Y_i$ represents document $i$'s keywords, as described in
Section~\ref{sub:designing_datashare}. We use $\GHash$ to map keywords to group
elements.

A naive approach to building \mspsi would be to mimic the client-server
protocols and to reuse the long-term key $\serverkey$ for all sets $\yseti{i}$.
This approach maps identical elements in sets $\yseti{i}, \yseti{j}$ to the same
tag revealing intersection cardinalities $| \yseti{i} \cap \yseti{j} |$.
We remove the link between tags across sets by adding a tag diversifying step to
the precomputation phase of client-server \psi (see
Figure~\ref{table:mspsi_proto}). We first compute pretags $\pretagset{i}$ for
each set $\yset_i$ by raising each element to the power of the long-term secret
$\serverkey$. Then, we compute per-set tags by hashing the pretags
$\clientPreTag$ with the set index $i$ to obtain $\hash(i \parallel
\clientPreTag)$. The hash-function ensures that the tags of each set are
independent. The server publishes the tag collection $\tagCollection$ and the
number of sets $\serversize$.

During the online phase, the client blinds its set as in the scheme of De
Cristofaro et al. and sends it to the server. The server re-blinds the set with
its secret $\serverkey$ and sends it back to the client in the same order. The
client unblinds the result to obtain the pretags for her elements. The client
then computes the corresponding tags $\clientTags^{(d)}$, for each document $d
\in [\serversize]$, and obtains the intersection.

\fullversioncmd{In Appendix~\ref{sec:formal_proof},} 
\confversioncmd{In the extended version~\cite{datashareextended} (Appendix A),} 
we prove the following theorem to show that the server learns nothing about the
client's set, and that the client learns nothing more than the intersections $\X
\cap \yseti{i}$.

\begin{theorem}
  \label{thm:mspsi-privacy}
  The \mspsi protocol is private against malicious adversaries in the random
  oracle model for $\hash$ and $\GHash$, assuming the one-more-gap
  Diffie-Hellman assumption holds.
\end{theorem}

The \mspsi protocol does not provide correctness against a malicious server, who
can respond arbitrarily leading the client to compute an incorrect intersection.
However, from Theorem~\ref{thm:mspsi-privacy} we know that, even then, the
malicious server cannot gain any information about the client's set.

\input{resource/protocol-mspsi.tex}

\para{Performance} 
Table~\ref{table:psi_performance} compares the performance of our \mspsi
protocol with the vanilla and the client-server \psi protocols in the multi-set
setting. We show the computation and communication cost for a server with
$\serversize$ sets and a client set with $\clientsize$ elements. \mspsi reduces
the server's online communication and computation by a factor $\serversize$. The
client can replace expensive group operations by inexpensive hash computations,
significantly reducing her online cost. The example costs for $\serversize =
1000$ (in square brackets) illustrate this reduction showing an improvement of 3
orders of magnitude.

\input{resource/psi-comparision-table.tex}

%% file: resource/notation-table.tex
\begin{table}[tb]
\centering
\caption{Notation.}
\label{table:notation}
\begin{tabular*}{\columnwidth}{@{}l@{\hskip4pt}l@{}}
	\toprule
	$\G, \gen, \grouporder$ & A cyclic group, its generator and the group's order\\ 
	$\secpar$ & The security parameter \\ 
	$x \randin X$ & Draw $x$ uniformly at random from the set $X$ \\ 
	$\hash, \GHash$ & Hash functions mapping into $\bin^{\secpar}$ resp.
                    group $\G$. \\
	$\NatNumUpTo{n}$ & The set $\{1, \ldots, n\}$ \\
	$\serverkey, \clientkey$ & The server's and client's secret keys\\	
	$\yset_i$ & The server's $i$th set $\YSETDef$\\	
	$\serversize, \setsize{i}$ & Nr. of server sets, resp. nr. of elements in set $\yset_i$\\	
	$\X$ & The client's set $\XDef$\\	
	$\clientsize$ & The number of elements in the client's set\\	
	$\clientPreTag, \pretagset{i}$ & Pretags for client
                                    ($\clientPreTag$) resp. the server's $i$th
                                    set $\yset_i$ ($\pretagset{i}$) \\	
	$\tagCollection$ & The server's tag collection \\	
	\bottomrule
\end{tabular*}
\end{table}

%% file: resource/protocol-psi-emiliano.tex
\begin{figure}[tb]
\centering
\begin{tabular}{@{}l@{\hskip1pt}c@{\hskip1pt}l@{}}
	\toprule
	\textbf{Client} & & \textbf{Server} \\
	\midrule
	$ \XDef \subset \G $ & & $ \yset = \set{\y_{1}, .., \y_{\setsize{}} } \subset \G $ \\
	\midrule
	$ \clientkey \randin \Zp$ & & $\serverkey \randin \Zp$ \\
	$ \xp_i = \x_i^{\clientkey} $ & \diagramsend{\langle \xp_i \rangle } & $ \xpp_i = \xp_i^{\serverkey}  $ \\
	$ \clientTags_i =  \hash(\xpp_i^{\clientkey^{-1}}\!)   $ & \diagramrecv{\langle \xpp_i \rangle ,\tagCollection} & $\tagCollection = \set{\hash(\y^{\serverkey}) \mid \y \in \yset} $ \\ 
	Return $\{ x_i \,|\, \clientTags_i \in \tagCollection \}$ & & \\
	\bottomrule
\end{tabular}
\caption{Vanilla \psi protocol by De Cristofaro et al. \cite{CristofaroGT12}.} 
\label{fig:psi-proto-emiliano}
\end{figure}

%% file: resource/protocol-mspsi.tex
\begin{figure}[tb]
\centering
\begin{tabular}{@{}l@{\hskip2pt}c@{\hskip2pt}l@{}}
	\toprule
	\textbf{Client} & & \textbf{Server} \\
	\midrule
	$ \XDef $ & & $ \makesetsub{\yset}{\serversize} $ \\
	$ $ & & $ \YSETDef $ \\
	\midrule
	\multicolumn{3}{c}{\emph{Precomputation phase}} \\
 	 & & $\serverkey \randin \Zp$ \\
	 & & $ \pretagset{i} = \set{ \GHash(\y)^{\serverkey} \mid \y \in \yset_{i}}$ \\ 

	 & \diagramrecv{\tagCollection, \serversize} & $\begin{aligned}
		\tagCollection = &\{\hash(i \concat t) \mid \\
				 & i \in \NatNumUpTo{\serversize} \land t \in \pretagset{i}\}
	\end{aligned}$ \\

	\midrule
	 & \emph{Online phase} & \\
	$ \clientkey \randin \Zp$ & &  \\
	$  \xp_i = \GHash(\x_i)^{\clientkey} $ & \diagramsend{\langle \xp_i \rangle} & $ \xpp_i = \xp_i^{\serverkey} $\\
	$  \clientPreTag_i = \xpp_i^{\clientkey^{-1}} $ & \diagramrecv{\langle \xpp_i \rangle} &  \\
  \textbf{For} $d \in \{1, \ldots, \serversize\}:$ \\
	\multicolumn{2}{@{}l}{$\quad \clientTags^{(d)}_i = \hash(d \concat \clientPreTag_i)$} \\
	\multicolumn{2}{@{}l}{Return $\{I_d = \{ x_i \,|\, \clientTags^{(d)}_i \in \tagCollection \} \}_{d \in [\serversize]}$} \\
	\bottomrule
\end{tabular}
\caption{Our \mspsi protocol. }
\label{table:mspsi_proto}
\end{figure}

%% file: resource/psi-comparision-table.tex
\begin{table}[tb]
\caption{Performance of PSI variants in a multi-set scenario: $\serversize$ is
	the number of server sets; $\totalsize$ is the total number of server elements;
	$\clientsize$ is the size of the client set; and \expCost{} and \hashCost{}
	denote the cost of an exponentiation and a hash computation ($\exphashCost{} =
	\hashCost{}+ \expCost{}$). We report in square brackets the cost estimation when
	$\clientsize=\sampleclientsize$, $\serversize=\samplesetsize$, $\totalsize =
	\sampletotalsize$ (i.e., server sets have $\sampleserversize$ elements). We
	assume that group elements require 32 bytes, $\expCost{}= 100\,\mu\textrm{s}$,
	and $\hashCost{}=1\,\mu\textrm{s}$.}
\label{table:psi_performance}
\centering
\begin{tabular}{lccc@{}}
	\toprule
	 & Vanilla  & \cspsi & \mspsi\\
	\midrule
	\multicolumn{4}{@{}l@{}}{\textit{Precomputation phase}}  \\
	Server & --- & $ \exphashCost{\totalsize}$ & $ \exphashCost{\totalsize}$ \\
	Comms & --- & $ \totalsize $ & $ \totalsize$ \\
	\midrule

	\multicolumn{4}{@{}l@{}}{\textit{Online phase}} \\
	Client & $\exphashCost{2\clientsize\serversize}$ & $\exphashCost{2\clientsize\serversize}$ & $ \expCost{2\clientsize}+\hashCost{\clientsize\serversize}$\\
	 & [2 s] & [2 s] & [12 ms] \\

	Server & $ \exphashCost{\totalsize} +\expCost{\clientsize\serversize}$ & $ \expCost{\clientsize\serversize}$ & $ \expCost{\clientsize}$ \\
	 & [11 s] & [1 s] & [1 ms] \\

	Comms & $\totalsize+2\clientsize\serversize$ & $ 2\clientsize\serversize $ & $ 2\clientsize$ \\
	 & [3.84 MB] & [640 KB] & [640 B] \\
	
	\bottomrule
\end{tabular}
\end{table}

%% file: parts/s4-messaging.tex
\section{Privacy-Preserving Messaging} 
\label{sec:messaging}
In this section, we introduce \name's communication system (\MIS). Journalists
use the \MIS to support \mspsi-based search and to converse anonymously after
they find a match. The \MIS respects the organization's limitations (see
Section~\ref{sec:system:requirements}). The communication costs do not hinder
the day-to-day operation of journalists, and the system supports asynchronous
communication. As the organization cannot deploy non-colluding nodes, the \MIS
uses one server. This server is trusted for availability, but not for privacy.

\name's communication system is designed to host short conversations for
discussing the sharing of documents. We anticipate that journalists will migrate
to using encrypted email or secure messengers if they need to communicate over a
long period or if they need to send documents.

\newcommand{\broadcastSend}{\textsf{BB.broadcast}\xspace}
\newcommand{\broadcastReceive}{\textsf{BB.read}\xspace}
\newcommand{\phclsetup}{\textsf{PH.setups}\xspace}
\newcommand{\phsend}{\textsf{PH.send}\xspace}
\newcommand{\phreceive}{\textsf{PH.recv}\xspace}
\newcommand{\phmailbox}{\textvar{mb}\xspace}
\newcommand{\phaddr}{\textvar{addr}\xspace}
\newcommand{\symkey}{k\xspace}
\newcommand{\lastseen}{\textvar{period}\xspace}
\newcommand{\dhkex}{k'\xspace}
\newcommand{\tdelay}{t_\textvar{delay}}
\newcommand{\delaydist}{\ensuremath{t_{\textvar{delay}}}\xspace}
\newcommand{\delaydmode}{\textvar{mode}\xspace}
\newcommand{\messagelen}{\textvar{mlen}\xspace}

\subsection{Messaging System Construction}
The server provides two components: a \emph{bulletin board} for broadcast
messages, and a \emph{pigeonhole} for point-to-point messages. We use
communication server to refer to the entity that operates both components. To
hide their network identifiers from the server and network observers,
journalists always use Tor\cite{DingledineMS04} for communication. To ensure
unlinkability, \name creates a new Tor circuit for every request.

\para{Bulletin Board} 
The bulletin board implements a database that stores broadcast messages.
Journalists interact with the bulletin board by using two protocols:
$\broadcastSend(m)$, which adds a message $m$ to the database to broadcasts it
to all journalists, and $m \gets \broadcastReceive()$ to retrieve unseen
messages.

\newcommand{\phsendraw}{\textsf{PH.SendRaw}\xspace}
\newcommand{\phrecvproc}{\textsf{PH.\-Recv\-Process}\xspace}
\newcommand{\phmonitor}{\textsf{PH.Monitor}\xspace}
\newcommand{\phcover}{\textsf{PH.Cover}\xspace}
\newcommand{\phsendhidden}{\textsf{PH.HiddenSend}\xspace}

\para{Pigeonhole} 
The pigeonhole consists of a large number of one-time-use mailboxes. Journalists
use the pigeonhole to send and receive replies to search queries and to
conversation messages. Journalists use the method \phsendraw
(Protocol~\ref{prot:ph:sendraw}) to send query replies; and the asynchronous
process \phrecvproc (Protocol~\ref{prot:ph:recvproc}) to retrieve incoming query
replies and conversation messages. Journalists use \phmonitor
(Protocol~\ref{prot:ph:monitor}) to receive notifications of new messages from
the pigeonhole and to trigger \phrecvproc.
Journalists are expected to connect to the system several times a week (see
Section~\ref{sec:system:requirements}). In agreement with ICIJ, we
decided that the pigeonhole will delete messages older than $7$ days.

Journalists may initiate a conversation after receiving a successful match. To
hide this event, we ensure that the \emph{sending of conversation messages is
unobservable}: the server cannot determine whether a journalist sends a
conversation message or not (see Definition~\ref{def:unobservability}). This
hides whether a conversation occurred, and therefore whether the search revealed
a match or not. To ensure unobservability of conversation messages, journalists
run \phcover (Protocol~\ref{prot:ph:cover}) to send cover messages at a constant
Poisson rate to \emph{every} journalist. To send a conversation message, it
suffices to replace one of the cover messages with the real message (see
\phsendhidden, Protocol~\ref{prot:ph:sendhidden}).

\newcommand{\authenc}{\textsf{AE}\xspace}
\newcommand{\authencenc}{\textsf{\authenc.enc}}
\newcommand{\authencdec}{\textsf{\authenc.dec}}

Journalists use the Diffie-Hellman key exchange to compute mailbox addresses and
message encryption keys, and an authenticated encryption scheme \authenc to
encrypt messages. 
Queriers generate a fresh key for every query and use that key to receive query
replies and to send conversation messages associated with that query. Document
owners use a medium-term key to send query replies and to receive conversation
messages from queriers (see Section~\ref{sub:designing_datashare}). When
exchanging cover traffic, journalists use fresh cover keys to send and their
medium-term keys to receive.

\begin{protocol}[$\phsendraw(\jourSK_S, \jourPK_R, m)$]
  \label{prot:ph:sendraw}
  To send message $m$ to recipient $R$ with public key $\jourPK_R$, a sender
  with private key $\jourSK_S$ proceeds as follows. Let $n_s$ be the number of
  times $S$ called $\phsendraw$ to send a message to $R$ before. The sender
  \begin{enumerate}[noitemsep,topsep=0pt]
    \item computes the Diffie-Hellman key $\dhkex = \textvar{DH}(\jourSK_S, \jourPK_R)$;
    \item computes the random rendezvous mailbox $\phaddr
      = \hash(\text{`addr'} \concat \dhkex \concat \jourPK_S \concat n_s)$ and a symmetric
      key $\symkey = \hash(\text{`key'} \concat \dhkex \concat \jourPK_S \concat n_s)$;
    \item pads the message $m$ to obtain $m'$ of length $\messagelen$, and
      computes the ciphertext $c = \authencenc(\symkey, m')$;
    \item 
      opens an anonymous connection to the pigeonhole and uploads $c$ to mailbox
      $\phaddr$.
  \end{enumerate}
  For every upload, the pigeonhole notifies all monitoring receivers (see
  \phmonitor below) that a message arrived at $\phaddr$.
\end{protocol}

\begin{protocol}[$\phrecvproc(\jourSK_R, \jourPK_S)$]
  \label{prot:ph:recvproc}
  To receive a message from sender $S$ with public key $\jourPK_S$, a receiver
  $R$ with private key $\jourSK_R$ runs the following asynchronous process. Let 
  $n_r$ be the number of times $R$ successfully received a message from $S$.
  The receiver
  \begin{enumerate}[noitemsep,topsep=0pt]
    \item computes the Diffie-Hellman key 
      $\dhkex = \textvar{DH}(\jourSK_R, \jourPK_S)$;
    \item uses $\dhkex$ to compute a random rendezvous mailbox 
      $\phaddr = \hash(\text{`addr'} \concat \dhkex \concat \jourPK_S \concat
      n_r)$ and a 
      symmetric key $\symkey = \hash(\text{`key'} \concat \dhkex \concat \jourPK_S \concat
      n_r)$;
    \item \label{step:ph_receive_wait} waits until \phmonitor (see below) receives
      a notification of a new message on address $\phaddr$.
      If no message is posted to $\phaddr$ in seven days, the process terminates;
    \item
      opens an
      anonymous connection to the pigeonhole and downloads the ciphertext $c$ at
      address $\phaddr$ (if there was no message due to a false positive, the
      process continues at step \ref{step:ph_receive_wait}); and
    \item decrypts the message $m' = \authencdec(\symkey, c)$ and returns the
      unpadded message $m$ or $\bot$ if decryption failed.
  \end{enumerate}
  When the receiver goes offline, this process is paused and resumed when the
  receiver comes online again.
\end{protocol}

A sender may send multiple messages without receiving a response.
The receiver calls $\phrecvproc$ repeatedly to receive all messages ($n_r$
increases every time). To ensure that the participants derive the correct
addresses and decryption keys, participants keep track of the message counters $n_s,
n_r$ for each pair of keys $(\jourSK_S, \jourPK_R)$ and $(\jourSK_R,
\jourPK_S)$, respectively.

\begin{protocol}[$\phmonitor$]
  \label{prot:ph:monitor}
  Journalists run the $\phmonitor$ process to monitor for incoming messages. The
  receiver
  \begin{enumerate}[noitemsep,topsep=0pt]
    \item \label{step:monitor:bulk}
      opens an anonymous monitoring connection to the pigeonhole and requests a list 
      of addresses $\phaddr$ that received a message since she was last online
    \item \label{step:monitor:feed}
      via the same anonymous connection, receives notifications of addresses
      $\phaddr$ with new messages.
  \end{enumerate}
  Addresses $\phaddr$ received in step~\ref{step:monitor:bulk}
  or~\ref{step:monitor:feed} can cause the \phrecvproc processes to continue
  past step~\ref{step:ph_receive_wait}. To save bandwidth, the pigeonhole sends
  a cuckoo filter~\cite{FanAK13} that contains the addresses in
  step~\ref{step:monitor:bulk}. Moreover, the pigeonhole only sends the first
  two bytes of the address in step~\ref{step:monitor:feed} (\phrecvproc handles
  false positives).
\end{protocol}

\newcommand{\ratecoverkey}{\lambda_k}
\newcommand{\tcoverkey}{t_{k}}
\newcommand{\expdelay}[1]{\textrm{Exp}(#1)}
\newcommand{\coversk}{\jourSK_c}
\newcommand{\coverpk}{\jourPK_c}
\newcommand{\trecipient}[1]{t_{#1}}
\newcommand{\ratecoversend}{\lambda_c}

The \phcover and \phsendhidden protocols ensure conversation messages are
unobservable. Senders store a queue of outgoing conversation messages for each
recipient.

\begin{protocol}[$\phcover(\jourSK_R)$]
  \label{prot:ph:cover}
  As soon as the journalists come online, they start the \phcover process. Let
  $\jourSK_R$ be the medium-term private key, and $\jourPK_1,\ldots,
  \jourPK_{n-1}$ be the medium-term public keys of the other journalists. The
  process runs the following concurrently:
  \begin{itemize}[noitemsep,topsep=0pt]
    \item \emph{Cover keys.} Draw an exponential delay
      $\tcoverkey \gets \expdelay{1/\ratecoverkey}$, and wait for time $\tcoverkey$.
      Generate a fresh cover key-pair $(\coversk, \coverpk)$ and upload
      $\coverpk$ to the bulletin board by calling $\broadcastSend(\coverpk)$.
      Repeat.
    \item \emph{Sending messages.}
      Wait until the first cover key has been uploaded. For each recipient
$\jourPK_i$, proceed as follows:
      \begin{enumerate}[noitemsep,topsep=0pt]
      \item Draw $\trecipient{i} \gets \expdelay{1/\ratecoversend}$ and wait
         for time $\trecipient{i}$.
       \item If the send queue for $\jourPK_i$ is not empty, let $m_i$ be the
         first message in the queue and $\jourSK_q$ the corresponding query key.
         Send the message by calling $\phsendraw(\jourSK_q, \jourPK_i, m_i)$ and
         remove $m_i$ from the queue. Otherwise, let
         $\coversk$ be the most recent private cover key and 
         $m_i$ be a dummy message.
         Send the message by calling $\phsendraw(\coversk, \jourPK_i, m_i)$.
        \item Repeat. 
       \end{enumerate}
     \item \emph{Receiving cover messages.} For each of the non-expired cover keys
       $\coverpk'$ on the bulletin board, call the process $m \gets \phrecvproc(\jourSK_R,
       \coverpk').$ If $m$ is a real message (see Section~\ref{sub:designing_datashare})
       forward the message to \name, otherwise discard. Repeat.
  \end{itemize}
  This process stops when the user goes offline, and $\phrecvproc$ processes 
  started by \phcover are canceled.
\end{protocol}

\begin{protocol}[$\phsendhidden(\jourSK_S, \jourPK_R, m)$]
  \label{prot:ph:sendhidden}
  To send a message $m$ to recipient $R$ with public key $\jourPK_R$, sender $S$ with
  private key $\jourSK_S$ places $m$ in the send queue for $\jourPK_R$.
\end{protocol}

\begin{figure*}[tbp]
  \centering
  \includegraphics[width=0.32\textwidth]{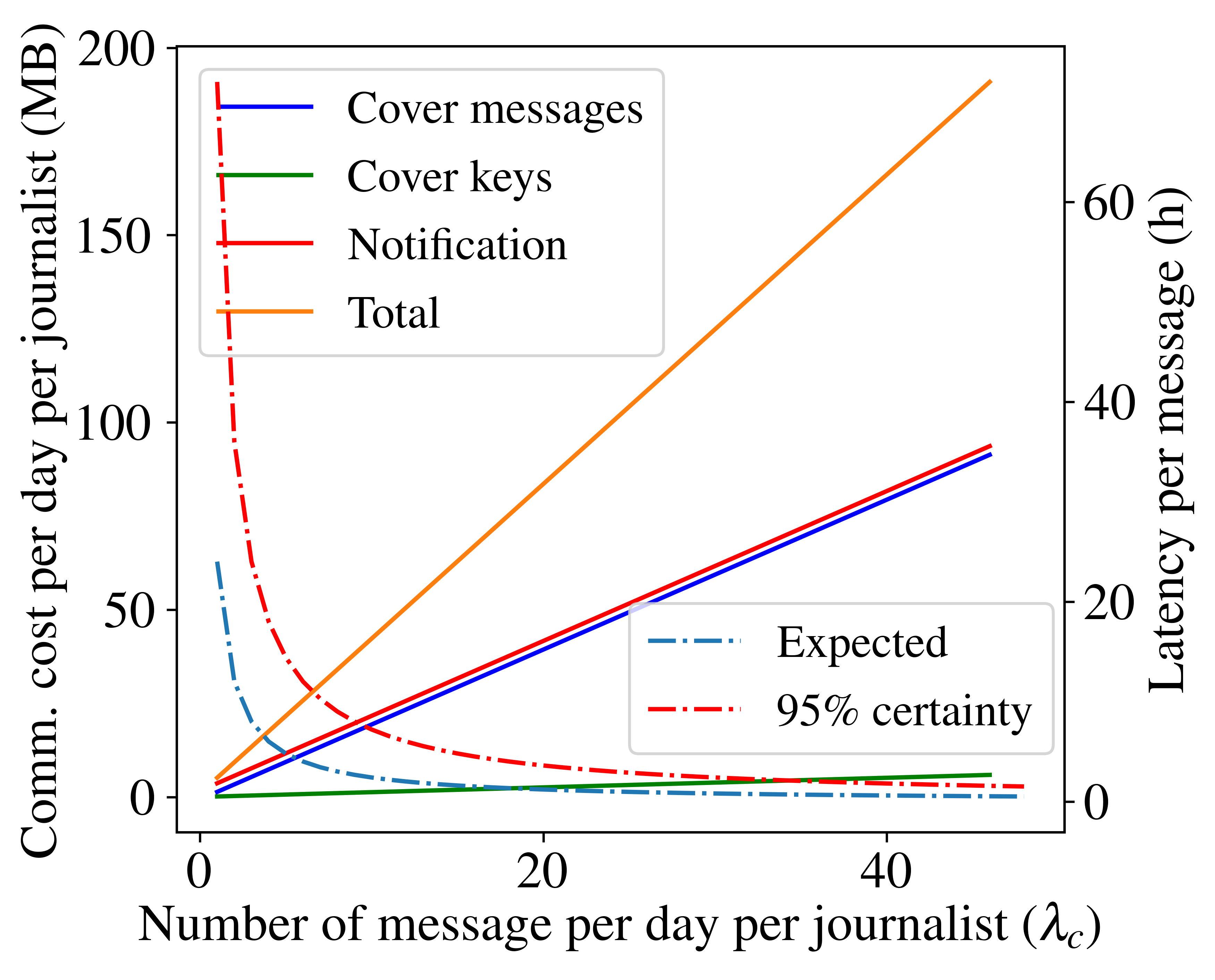}
  \hfill
  \includegraphics[width=0.30\textwidth]{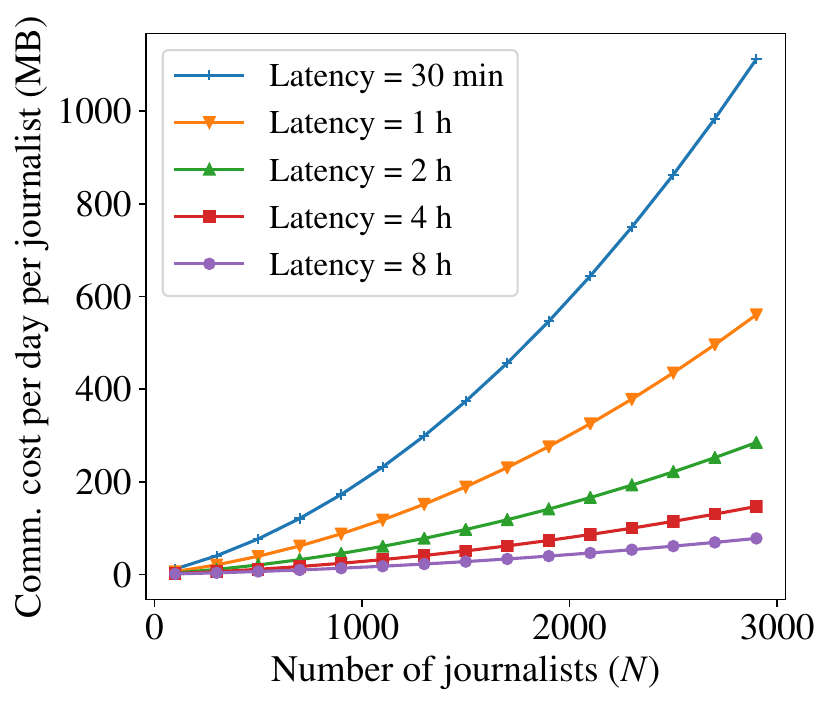}
  \hfill
  \includegraphics[width=0.32\textwidth]{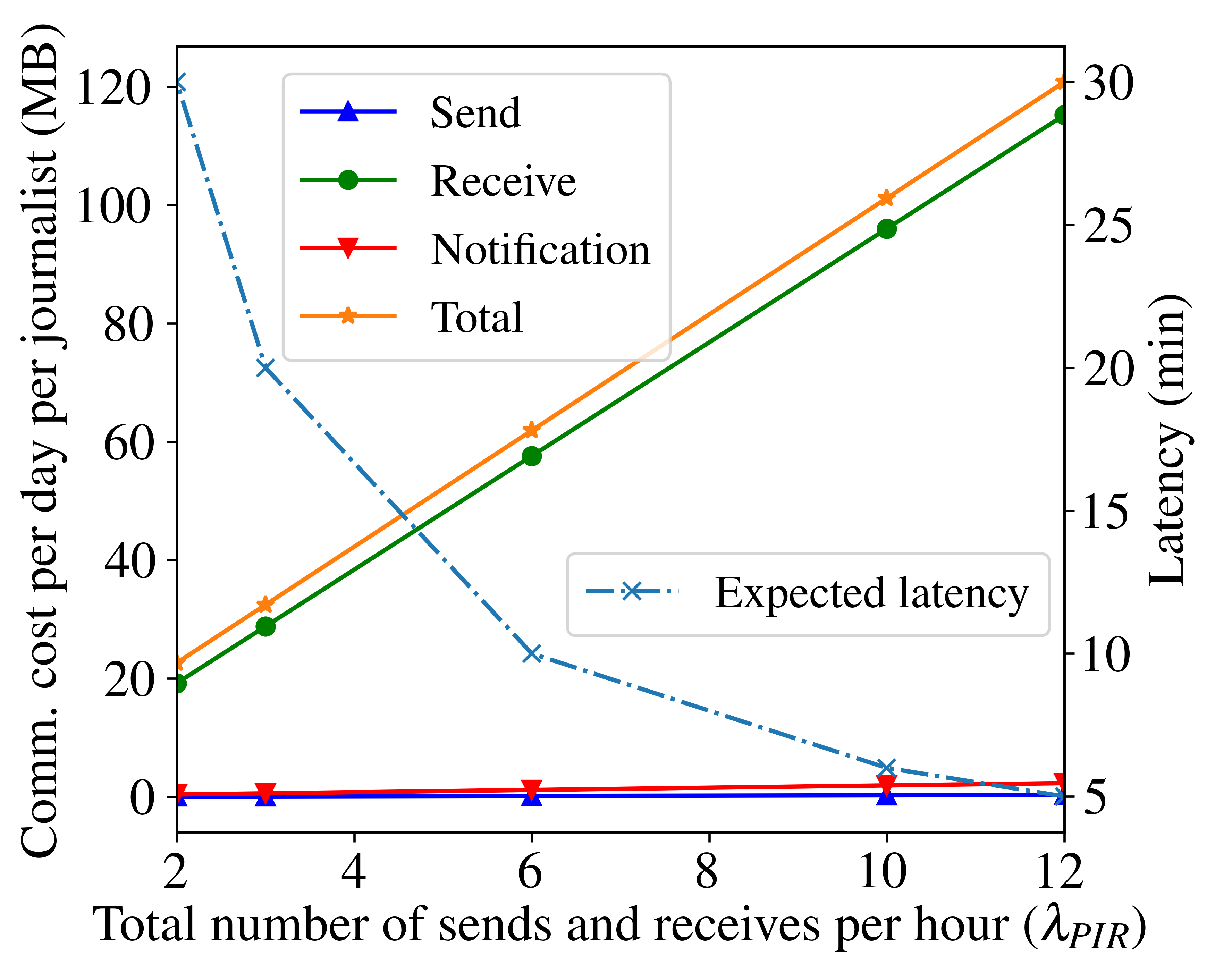}
  \caption{Left: bandwidth (left axis) and latency (right axis) for running the
    communication system (CS) with 1000 journalist for given rate
    $\ratecoversend$. Middle: varying the number of journalists and average
    latency in the CS. Right: bandwidth (left axis) and latency (right axis) for
    running the PIR system with 1000 journalists.}
  \label{fig:messaging_cost}
\end{figure*}

\subsection{Messaging Service Privacy}
\label{sub:messaging_privacy}
We first define unobservability then prove that conversation messages sent using
\phsendhidden are unobservable.

\begin{definition}[Unobservability]
  \label{def:unobservability}
  A conversation message is unobservable if all PPT adversaries have a negligible
  advantage in distinguishing a scenario in which the sender $S$ sends a
  conversation message to the receiver $R$, from a scenario where $S$ does not
  send a conversation message to $R$.
\end{definition} 

\begin{theorem}
  \label{thm:unobservable}
  Messages sent using \phsendhidden are unobservable towards any adversary that
  controls the communication server but does not control the sender or the
  receiver, assuming the receiver awaits both conversation and cover messages.
  This statement is also true when the adversary can break the network anonymity
  Tor provides.
\end{theorem}

\begin{proof}
  To show that conversation messages are unobservable, we must prove that the
  following two scenarios are indistinguishable: the scenario in which the
  sender sends a conversation message (sent by \phcover after a conversation
  message has been queued using \phsendhidden), and the scenario in which the
  sender sends a cover message (sent by \phcover when no conversation message
  has been queued). The intuition behind this proof is that the conversation and
  cover messages are indistinguishable: (1) both are encrypted so that the
  adversary cannot distinguish them based on content; and (2) conversation
  messages replace cover messages, so they are sent using the same schedule.

  All messages go through the pigeonhole. For each message, the adversary
  observes (1) the pigeonhole address, (2) the content, (3) the length, (4) the
  timestamps at which the message was posted and retrieved, and 
  -- in the worst case scenario in which the adversary can break
  the anonymity Tor provides -- (5) the sender and the receiver. 
    
  The content and pigeonhole address of messages are cryptographically
  indistinguishable. Senders and receivers compute rendezvous mailbox addresses
  by using a Diffie-Hellman key exchange based on either the query public key
  and the owner's public key (when the message is a conversation messages) or
  the sender and receiver's cover keys (when the message is a cover message). As
  the adversary does not control the sender or the receiver, it does not know
  the corresponding private keys in either scenario. Under the decisional
  Diffie-Hellman assumption, the adversary cannot distinguish between mailbox
  addresses for conversation messages and mailbox addresses for cover messages.
  Under the same DH assumption, the adversary cannot learn the symmetric key $k$
  that is used to encrypt the message either. Moreover, all messages are padded
  to a fixed length of $\messagelen$. Hence, the adversary cannot distinguish
  between the two situations based on message content or length.
  As a result, all messages sent between sender $S$ and receiver $R$ are
  indistinguishable to the adversary on the cryptographic layer. 

  We now show that the post and retrieve times of the messages are also independent of
  whether the message is a cover message or a conversation message:

  \noindent\emph{Sender}. The ``cover keys'' and ``sending messages'' processes
  of \phcover are, by design, independent of whether a conversation message
  should be sent or not. 
  The sender sends
  (real or cover) messages to the recipient at a constant rate $\ratecoversend$.
  The send times are independent of whether the sender has a real message for
  the receiver.

  \noindent\emph{Receiver}. The receiver is listening to both conversation and
  cover messages from the sender. As soon as it a new message notification arrives, 
  \phrecvproc will retrieve this message. Therefore, the retrieval time 
  does \emph{not} depend on the type of message.
\end{proof}

As a corollary of the unobservability proof, we have the following theorem.
\begin{theorem}
  \label{thm:secrecy}
  The pigeonhole protects the secrecy of messages from non-participants
  including the communication server.
\end{theorem}

To hide their (network) identities from the communication server, users of \name
communicate with the communication server via Tor. Sender anonymity hides
queriers' identities from document owners, and receiver anonymity hides document
owners' identities from queriers. Using Tor ensures these properties, even when
journalists collude with the communication server. Formally, we define sender
and receiver anonymity as follows:

\begin{definition}[Sender anonymity]
  \label{def:sender_anon}
  A communication system provides sender anonymity if any PPT adversary has a
  negligible advantage in guessing the sender of a message.
\end{definition} 

\begin{definition}[Receiver anonymity]
  \label{def:receiver_anon}
  A communication system provides receiver anonymity if any PPT adversary has a
  negligible advantage in guessing the receiver of a message.
\end{definition} 

\begin{theorem}
  \label{thm:anonymity}
  Assuming that Tor provides sender and receiver anonymity with respect to the
  communication server, the communication system provides sender and receiver
  anonymity at the network layer against adversaries who control the
  communication server and a subset of journalists.
\end{theorem}

\begin{proof}
  All messages go through the communication system and journalists never
  directly connect with each other. We study separately the anonymity provided
  by the bulletin board and the pigeonhole.

  To publish an encrypted message (the query) to the bulletin board, senders run
  the \broadcastSend protocol over a fresh Tor circuit. Sender anonymity is
  guaranteed by Tor. The bulletin board broadcasts all messages to all
  journalists. As these messages do not have an intended receiver, receiver
  anonymity is not relevant. 

  Both senders and receivers use fresh Tor circuits when communicating with the
  communication servers. This ensures that communications are unlinkable at the
  network layer, and that the adversary cannot identify the journalist from
  network artifacts. As shown in the unobservability proof, the pigeonhole
  cannot distinguish senders' or receivers' given addresses or encrypted
  messages.
\end{proof}

This theorem only addresses the anonymity at the network layer. We discuss
anonymity at the application layer, i.e., based on the content of messages, in
Section~\ref{sub:leakage_analysis}.

Tor does not provide sender or receiver anonymity against global passive
adversaries. To protect against global passive adversaries, \name will migrate
to stronger network layer anonymity systems (e.g., the Nym system~\cite{Nym},
based on Loopix~\cite{PiotrowskaHEMD17}) 

\subsection{Cost Evaluation}
\label{sec:mess_eval}
To guarantee unobservability, we schedule the traffic according to a Poisson
distribution. However, such strong protection comes at a cost~\cite{Das0MK18}:
Regardless of whether they have zero, one, or many conversations, every
journalist sends messages at a rate $\ratecoversend$ to the other $N$
journalists, i.e., sends $\ratecoversend N$ messages per day. Consequently,
every journalist also receives $\ratecoversend N$ messages a day.

Figure~\ref{fig:messaging_cost}, left, illustrates the trade-off between
bandwidth overhead and latency for a given cover traffic rate. When journalists
send few messages a day, the bandwidth requirements are very low. For instance,
setting $\lambda_c$ to be 4 messages per day requires every journalist to use
16.5\,MB per day, including the sending of notifications and the updating of
cover keys. For these messages to be unobservable, however, journalists have to
wait on average six hours between messages (less than 18 hours in 95\% of the
cases).
If journalists require higher throughput they must consume more bandwidth. For
example, setting $\lambda_c=48$ messages a day ensures that messages are sent
within half an hour on average (and within 90 minutes with probability 95\%).
Storing messages from the last seven days on the pigeonhole for $1000$
journalists and send rate of $\lambda_c=48$ requires $390$\,GB, which is
manageable for a server.

The latency we report in Figure~\ref{fig:messaging_cost} assumes that
journalists are online. If they disconnect from the system before a message is
sent, journalists must, after coming online again, first upload a new cover key
then draw a new sample from $\expdelay{\lambda_c}$ to decide when to send their
message. We propose to set the update latency $\ratecoverkey$ to $\ratecoversend
/ 4$, so that the initial latency is at most 25\% more than the latency under
normal circumstances. 

For the current size of the population that will use \name, 250 journalists (see
Section~\ref{sec:system:requirements}), the bandwidth can be kept reasonable at
the cost of latency. However, as journalists send cover traffic to everyone, the
bandwidth cost increases quadratically with the size of the population, and
becomes pretty heavy after reaching 2000 journalists, see
Figure~\ref{fig:messaging_cost}, center.

\newcommand{\ratecoverpir}{\lambda_{\textrm{PIR}}}
\para{An Alternative Construction}
If the traffic requirements become too heavy for the organization members,
bandwidth can be reduced by increasing the computation cost at the pigeonhole
server. Instead of using cover traffic to \emph{all} journalists to hide the
mailboxes that contain real messages, journalists can retrieve messages using
computational private information retrieval
(PIR)~\cite{KushilevitzO97,AngelCLS18}.

In this approach, senders send cover messages at a rate $\ratecoverpir$,
\emph{independent} of the number of journalists, to random mailboxes. When they
have a real message, they send it instead of a cover message. They use the same
rate to retrieve messages using PIR. This approach hides which messages are
getting retrieved from the pigeonhole and breaks the link between the send and
receive time. As a result, the server's observation of the system is independent
of whether journalists send a real message or not. 

We illustrate the trade-off associated with this approach in
Figure~\ref{fig:messaging_cost}, right. We use SealPIR~\cite{AngelCLS18} to
retrieve cover and conversation messages. Responding to a PIR request in a
scenario of 1000 journalists and a send rate of 6 messages per hour takes $12$
seconds. Therefore, we assume a server with 24 cores (approx 1300 USD/month in
AWS) can handle this scenario. We see that this approach enables the system to
send conversation messages at a higher rate and a lower cost. For example,
sending 6 messages per hour (144 messages a day) requires around 59\,MB.
However, as opposed to the Poisson cover approach described in the previous
section, this rate limits the total number of messages per day \emph{regardless
of recipient}. As a result, depending on the number of receivers journalists
want to communicate with on average, one or the other method could be more
advantageous.

%% file: parts/s5-datashare.tex

\section{The \DataShare System} 
\label{sec:the_datashare_system}
We now present \name, an asynchronous decentralized peer-to-peer document search
engine. \name combines the multi-set private set intersection protocol
(Section~\ref{sec:multi_set_psi}), the privacy-preserving communication system
(Section~\ref{sec:messaging}), and an anonymous authentication mechanism.

\newcommand{\bs}{\textsf{BS}\xspace}
\newcommand{\bssetup}{\textsf{BS.Setup}\xspace}
\newcommand{\bssign}{\textsf{BS.Sign}\xspace}
\newcommand{\bsverify}{\textsf{BS.Verify}\xspace}

\subsection{Preliminaries} 
\label{sub:preliminaries}
\para{Processing Documents}
The primary interests of investigative journalists are named entities, such as
people, locations, and organizations (see Section
\ref{sec:system:requirements}). \ICIJ has already developed a tool
\cite{DatasharePlatform} that uses natural language processing to extract named
entities from documents. After the extraction, the tool transforms named
entities into a canonical form to reduce the impact of spelling variation in
names. We employ this tool to canonicalize queries. An advantage of using this
tool over simply listing all words in a document is that it reduces the number
of keywords per document: the majority of documents have less than 100 named
entities.

\para{Search}
\name uses the \mspsi protocol as a pairwise search primitive between
journalists. The querier acts as \mspsi client, and the client's set represents
the querier's search keywords. The document owners act as \mspsi servers, where
the server's $N$ sets represent the keywords in each of the owner's $N$
documents. Each document owner has their own \emph{different} corpus and secret
key. We say a document is a match if it contains \emph{all} query keywords
(i.e., the conjunction of the query keywords, see
Section~\ref{sec:system:requirements}). \mspsi speeds up the computation and
reduces the communication cost by a factor of $N$ compared to the naive approach
of running one PSI protocol per document.

\para{Authenticating Journalists}
Only authorized journalists, such as members of the organization or
collaborators, are allowed to make queries and send conversation messages.
\name's authentication mechanism operates in epochs. In each epoch journalists
obtain a limited number of anonymous tokens. Tokens can be used only once, which
limits the number of queries that journalists can make per epoch. Compromised
journalists, therefore, can extract limited information from the system by
making search queries. We considered using identity-escrow mechanisms to
mitigate damage by misbehaving journalists but in agreement with the
organization, we decided against this approach as such mechanisms could too
easily be abused to identify honest journalists.

Recall from Section~\ref{sec:system:requirements} that journalists trust the
organization as an authority for membership and already have means to
authenticate themselves to the organization. Therefore, the organization is the
natural design choice for issuing anonymous tokens. We note that, even if the
organization is compromised, it can do limited damage as it cannot link queries
or conversations to journalists (because of token anonymity). However, it can
ignore the rate limit. This would enable malicious queriers to extract more
information than allowed. To mitigate this risk, \name could also work with
several token issuers and require a threshold of valid tokens.

For the epoch duration, \ICIJ proposes one month to provide a good balance
between protection and ease of key management. Rate-limits are flexible. The
organization can decide to provide additional one-time-use tokens to journalists
who can motivate their need for extra tokens. Although this reveals to the
organization which journalists are more active, it does not reveal what they use
the tokens for.

\parait{Instantiation.} 
Tokens take the form of a blind signature on an ephemeral signing key. We use
Abe's blind signature (BS) scheme~\cite{Abe01}. The organization runs
$\bssetup(1^\secpar)$ to generate a signing key $\gsmsk$ and a public
verification key $\gsmpk$. To sign an ephemeral key $\jourPK_T$, the journalist
and the organization jointly run the $\bssign()$ protocol. The user takes as
private input the key $\jourPK_T$, and the organization takes as private input
its signing key $\gsmsk$. The user obtains a signature $\cred$ on $\jourPK_T$.
The verification algorithm $\bsverify(\gsmpk, \cred, \jourPK_T)$ returns $\top$
if $\cred$ is a valid for $\jourPK_T$ and $\bot$ otherwise. These blind
signatures are anonymous. The blindness property of \bs ensures that the signer
cannot link the signature $\cred$ or the key $\jourPK_T$ to the journalist that
ran the corresponding signing protocol.

Let $\jourSK_T$ be the private key corresponding to $\jourPK_T$. We call the tuple $T =
(\jourSK_T, \cred)$ an authentication token. Journalists use tokens to 
authenticate themselves before issuing a query or sending a
message. To authenticate themselves, journalists create a signature $\sigma$ on
the message using $\jourSK_T$ and append the signature $\sigma$ and blind
signature $\cred$ on $\jourPK_T$. Non-authenticated messages and queries are dropped by other journalists. 

Anonymous authentication with rate limiting could have been instantiated
alternatively with $n$-times anonymous credentials~\cite{CamenischHKLM06},
single show anonymous credentials~\cite{Brands00, BaldimtsiL13}, or regular
anonymous credentials~\cite{AuSMC13,PointchevalS16} made single-show. We opted
for the simplest approach.

\para{Cuckoo Filter} 
\label{par:cuckoo_filter}
\name uses cuckoo filters~\cite{FanAK13} to represent tag collections in a
space-efficient manner. The space efficiency comes at the price of having false
positives when answering membership queries. The false negative ratio is always
zero. The false positive ratio is a parameter chosen when instantiating the
filter. Depending on the configuration, a cuckoo filter can compress a set to
less than two bytes per element regardless of the elements' original size.

Users call $\cfcompress(\cfinputset, \cfparams)$ to compute a cuckoo filter
$\cuckoo$ of the input set $\cfinputset$ using the parameters specified in
$\cfparams$. Then, $\cfcheck(\cuckoo, \cfelem)$ returns $\top$ if $\cfelem$ was
added to the cuckoo filter, and $\bot$ otherwise. For convenience, we write
$\cfintersection(\cuckoo, \cfinputset')$ to compute the intersection
$\cfinputset' \cap \cfinputset$ with the elements $\cfinputset$ contained in the
cuckoo filter. The function $\cfintersection$ can be implemented by running
$\cfcheck$ on each element of $\cfinputset'$.


\newcommand{\systemsetup}{\textsf{SystemSetup}\xspace}
\newcommand{\journalistsetup}{\textsf{JournalistSetup}\xspace}
\newcommand{\gettoken}{\textsf{GetToken}\xspace}
\newcommand{\publish}{\textsf{Publish}\xspace}
\newcommand{\querycmd}{\textsf{Query}\xspace}
\newcommand{\reply}{\textsf{Reply}\xspace}
\newcommand{\process}{\textsf{Process}\xspace}
\newcommand{\contact}{\textsf{Converse}\xspace}
\newcommand{\cover}{\textsf{Cover}\xspace}
\newcommand{\verifycontact}{\textsf{VerifyContact}\xspace}

\subsection{\name Protocols and Design} 
\label{sub:designing_datashare}
The journalists' organization sets up the \name system by running \systemsetup
(Protocol~\ref{prot:setup}). Thereafter, journalists join \name by running
\journalistsetup (Protocol~\ref{prot:jsetup}). Journalists periodically call
\gettoken (Protocol~\ref{prot:gettoken}) to get new authentication tokens, and
\publish (Protocol~\ref{prot:publish}) to make their documents searchable. \name
does not support multiple devices, and the software running on journalists' machines
automatically handles key management without requiring human interaction. If a
journalist's key is compromised, she contacts the organization to revoke it.

\begin{protocol}[\systemsetup] \label{prot:setup}
  The journalist organization runs \systemsetup to set up the \name system:
  \begin{enumerate}[noitemsep,topsep=0pt]
    \item The organization  generates a cyclic group $\G$ of prime order 
      $\grouporder$ with generator $\gen$, and hash functions $\hash : \bin^*
      \rightarrow \bin^{\secpar}$ and $\GHash : \bin^* \rightarrow \G$ for use in
      the \mspsi protocol. It selects parameters $\cfparams$ for the cuckoo
      filter and sets the maximum number of query keywords $\kwlim$ (we use
      $\kwlim = 10$). The organization publishes these.
    \item The organization sets up a token issuer by running 
      $(\gsmsk, \gsmpk) = \bssetup(1^{\secpar})$ and publishes $\gsmpk$.
    \item The organization sets up a \DSpigeonhole, which provides a bulletin 
      board and a pigeonhole.
  \end{enumerate}
\end{protocol}

\begin{protocol}[\journalistsetup] \label{prot:jsetup} 
  Journalists run \journalistsetup to join the network: The journalist
  authenticates to the organization and registers for \name.
\end{protocol}

\begin{protocol}[\gettoken] \label{prot:gettoken}
  Journalists run \gettoken to obtain one-time-use authentication tokens from
  the organization.
  \begin{enumerate}[noitemsep,topsep=0pt]
    \item The journalist $\jour$ connects to the organization and authenticates
      herself. The organization verifies that $\jour$ is allowed to obtain an
      extra token and, if not, aborts.
    \item The journalist generates an ephemeral signing key 
      $(\jourSK_T, \jourPK_T)$; runs the $\bssign()$ protocol with the
      organization to obtain the organization's signature $\cred$ on the
      message $\jourPK_T$ (without the organization learning $\jourPK_T$); 
      and stores the token $T = (\jourSK_T, \cred)$.
  \end{enumerate}
\end{protocol}

To obtain tokens for the new epoch, journalists repeatedly run the \gettoken
protocol at the beginning of each epoch.

\begin{protocol}[\publish] \label{prot:publish}
  Journalists run \publish to make their documents searchable. \publish takes as input
  a token $T = (\jourSK_T, \cred)$ and a set $\textvar{Docs} = \{\doc_1, ..,
  \doc_{\serversize}\}$ of $\serversize$ documents such that each document
  $\doc_i$ is a set of keywords in $\{0,1\}^*$. This protocol includes the
  pre-computation phase of \mspsi.
  \begin{enumerate}[noitemsep,topsep=0pt]
    \item The journalist chooses a secret key $\serverkey \randin \Zp$ and
      computes her tag collection for the \mspsi protocol as 
      $$\tagCollection = \{\hash(i \concat \GHash(\y)^{\serverkey}) \mid 
      i \in \NatNumUpTo{\serversize},\; \y \in \doc_i \},$$
      and compresses it into a cuckoo filter
      $\cuckoo = \cfcompress(\tagCollection,\cfparams)$.
    \item The journalist generates a long-term pseudonym $\pid$, and a
      medium-term contact key pair $(\jourSK, \jourPK)$.
    \item The journalist encodes her pseudonym $\pid$, public key $\jourPK$,
      compressed tag collection $\cuckoo$, and the number of documents
      $\serversize$ as her public record
      $$\jRecord = (\pid, \jourPK,\cuckoo, \serversize).$$
    \item The journalist signs her record $\signature = \gssign(\jourSK_T,
      \jRecord)$ and runs $\broadcastSend(\jRecord \concat \signature \concat
      \jourPK_T \concat \cred)$ to publish it.
  \end{enumerate}
\end{protocol}

\name automatically rotates (e.g., every week) the medium-term contact key of
journalists $(\jourSK, \jourPK)$ to ensure forward secrecy. This prevents that
an attacker that obtains a journalist's medium-term private key can recompute
the mailbox addresses and encryption key of messages sent and received by the
compromised journalist.

Journalists retrieve all public records from the bulletin board. They run
$\gsverifysig(\jourPK_T, \signature, \jRecord)$ to verify the records against
the ephemeral signing key, check that they have not seen $\jourPK_T$ before
to enforce the one-time use, and run $\bsverify(\jourPK_T, \cred, \gsmpk)$ to
validate the blind signature. Journalists discard invalid records.

\begin{figure}[t]
    \includegraphics[width=\linewidth, trim=12mm 39mm 2mm 0mm,clip]{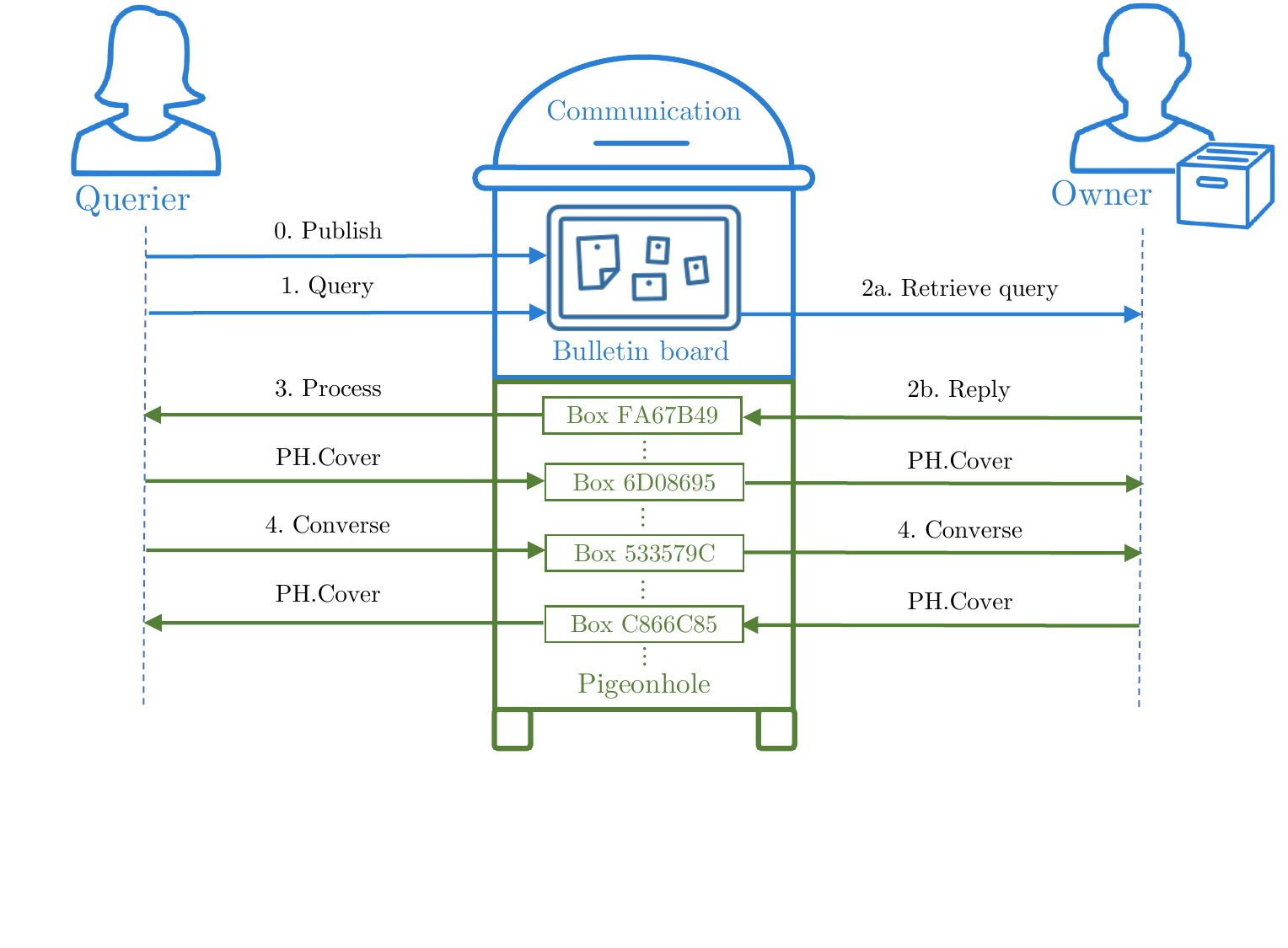}
    \caption{An overview of \DataShare protocols.}
    \label{fig:protocol_overview}
\end{figure}

\name incorporates \mspsi into its protocols to enable document search. Querying
works as follows (Fig.~\ref{fig:protocol_overview}): (1) The querier posts a
query together with a fresh key $\jourPK_q$ to the bulletin board
(Protocol~\ref{prot:query}); (2) Document owners retrieve these queries from the
bulletin board (2a), they compute the reply address, and they send the reply to
a pigeonhole mailbox (2b, see Protocol~\ref{prot:reply}); (3) The querier
monitors the reply addresses for all document owners, retrieves the replies, and
computes the intersection to determine matches (Protocol~\ref{prot:process}).

\begin{protocol}[\querycmd]\label{prot:query}
  Queriers run \querycmd to search for keywords $X$. The protocol takes as input
  a token $T = (\jourSK_T, \cred)$.
  \begin{enumerate}[noitemsep,topsep=0pt]
    \item The querier generates a key pair $(\jourSK_q, \jourPK_q)$ for the query
      and pads $X$ to $\kwlim$ keywords by adding random elements.
    \item As in the \mspsi protocol, the querier picks a fresh blinding factor 
      $\clientkey \randin \Zp$, and computes:
      \[\query = \{\GHash(\x)^\clientkey \setdelim \x \in X\}.\]
    \item The querier signs the query $\query$ and her public key $\jourPK_q$
      as $\signature = \gssign(\jourSK_T, \query \concat \jourPK_q)$, and
      broadcasts the query $\query$, public key $\jourPK_q$, signature
      $\signature$, ephemeral token key $\jourPK_T$, and token $\cred$ by running
      $\broadcastSend(\query \concat \jourPK_q \concat \signature \concat
      \jourPK_T \concat \cred)$. 
  \end{enumerate}
\end{protocol} 

Recall that \mspsi perfectly hides the keywords inside queries. As a result,
these queries can be safely broadcasted.

\newcommand{\sendmsg}{\textsf{SendMsg}\xspace}
\newcommand{\recvmsg}{\textsf{RecvMsg}\xspace}

\begin{protocol}[\reply]\label{prot:reply}
  Document owners run \reply to answer a query $(\query, \jourPK_q, 
  \signature, \jourPK_T, \cred)$ retrieved from the bulletin board.
  \begin{enumerate}[noitemsep,topsep=0pt]
  \item The owner verifies the query by checking $\gsverifysig(\jourPK_T,
    \signature, \allowbreak \query \concat\allowbreak \jourPK_q)$,
    $\bsverify(\gsmpk, \cred, \jourPK_q)$, and that she did not see $\jourPK_T$
    before. If any verification fails, she aborts. 
    \item The owner uses her secret key $\serverkey$ to compute the \mspsi
      response  $\response = \{\xp^{\serverkey} \setdelim \xp \in \query\}$ to the
      query.
    \item Let $\jourSK$ be the owner's medium-term private key. She runs
      $\phsendraw(\jourSK, \jourPK_q, \response)$ to post the result to the
      pigeonhole, and starts the process $\phrecvproc(\jourSK, \jourPK_q)$ to
      await conversation messages from the querier (see \contact below).
  \end{enumerate}
\end{protocol}

\begin{protocol}[\process]\label{prot:process}
  Queriers run the \process protocol for every journalist $\jour$ with record
  $\jRecord =(\pid, \jourPK, \cuckoo, \serversize)$ to retrieve and process
  responses to their query $(X, \jourSK_q, \clientkey),$ where $X$ is the
  unpadded set of query keywords.
  \begin{enumerate}[noitemsep,topsep=0pt]
    \item The querier runs the asynchronous protocol
    $\response \gets \phrecvproc(\jourSK_q, \jourPK)$ to get the new response.
    \item Similar to \mspsi, the querier computes the size of the intersection $I_i$
      for each document $d_i$, $1 \leq i \leq \serversize$, as 
      \[ I_i = \left| \cfintersection\left(\cuckoo, \{\hash( i \parallel
      \xpp^{\clientkey^{-1}})\setdelim \xpp \in \response\} \right) \right|.\] 
    \item Let $q = |X|$ be the number of query keywords. The querier learns that
the owner $\pid$ has $\docMatch = \left| \{i \;|\; I_i = q \} \right|$ matching
documents.
  \end{enumerate}
\end{protocol}

\newcommand{\timing}{\textvar{t}\xspace}
\newcommand{\delaydummy}{t_\textsf{dummy}\xspace}
\newcommand{\coverjour}{\textvar{D}\xspace}
\newcommand{\coverField}[1]{\field{\coverjour}{{#1}} }

After finding a match, the querier and owner can converse via the pigeonhole to
discuss the sharing of documents using the $\contact$ protocol.

\begin{protocol}[\contact] \label{prot:contact}
  Let $(\jourSK_q, \jourPK_q)$ be the query's key pair, and $(\jourSK_O,
  \jourPK_O)$ the owner's medium-term key pair at the time of sending the query.
  \begin{itemize}[noitemsep,topsep=0pt]
    \item The querier sends messages $m$ to the owner by calling
      $\phsendhidden(\jourSK_q, \jourPK_O, m)$, and awaits replies by calling
      $\phrecvproc(\jourSK_q, \jourPK_O).$
    \item The owner sends messages $m$ to the querier by calling
      $\phsendhidden(\jourSK_O, \jourPK_q, m)$, and awaits replies by calling
      $\phrecvproc(\jourSK_O, \jourPK_q)$.
    \item After receiving a message, the receiving party calls $\phrecvproc$
      again, to await further messages.
  \end{itemize}
\end{protocol}

Both the query's key $\jourPK_q$ and the owner's key $\jourPK_O$ are signed
using a one-time-use token. Thus, querier and owner know they communicate with
legitimate journalists.


\newcommand{\sendtime}[1]{\textvar{ST}_{{#1}}\xspace}
\newcommand{\receivetime}[1]{\textvar{RT}_{{#1}}\xspace}
\newcommand{\queryperiod}{\textvar{qp}\xspace}

\subsection{\name Security Analysis} 
\label{sub:leakage_analysis}
\name provides the following guarantees: 

\para{Protecting Queries}
The requirements established in Section~\ref{sec:system:requirements} state that
\name must protect the searched keywords and identity of the querier from
adversaries that control the communication server and a subset of document
owners. The \querycmd protocol, which handles sending queries, is based on
\mspsi. It represents searched keywords as the client's set in \mspsi.
Theorem~\ref{thm:mspsi-privacy} states that \mspsi perfectly hides the client's
set from malicious servers. Therefore, \name protects the content of queries
from owners.

\name does not reveal any information about the identity of queriers at the
network and application layer. Theorem~\ref{thm:anonymity} ensures that the
communication system provides sender and receiver anonymity and protects the
querier's identity at the network layer. At the application layer, the querier
sends $(\query \concat \jourPK_q \concat \signature \concat \jourPK_T \concat
\cred)$ as part of the \querycmd protocol to the bulletin board. The values
$\signature$, $\jourPK_T$, and $\cred$ form an anonymous authentication token
based on Abe's blind signature \cite{Abe01}. Anonymous tokens are independent of
the querier's identity. The value $\jourPK_q$ is an ephemeral public key, and
$\query$ is a \mspsi query which uses an ephemeral secret for the client. Hence,
both $\jourPK_q$ and $\query$ are independent of the querier's identity too.
Therefore, the content of the query does not leak the querier's identity at the
application layer.

\para{Protecting Conversations}
According to the requirements stated in Section~\ref{sec:system:requirements}, \name
must protect (1) the content, and (2) the identity of participants in a
conversation from non-participants. (3) \name must protect the
identities of journalists (who are in a conversation) from each other.

First, \name protects the content of conversation messages from
non-participants: Theorem~\ref{thm:secrecy} proves that
only the sender and receiver can read their conversation messages.

Second, \name protects the identity of participants in a conversation from
non-participants. Theorem~\ref{thm:unobservable} proves that communication is
unobservable, as long as participants are awaiting both conversation and cover
messages. \name enforces the conditions by construction. Immediately after
answering a query (see \reply, Protocol~\ref{prot:reply}), the owner starts
\phrecvproc to listen for messages from the querier. Similarly, the querier
starts to listen for conversation messages from the owner right after sending
him a conversation message (see \contact, Protocol~\ref{prot:contact}).
Moreover, the ``cover keys'' and ``receiving cover messages'' processes in the
\phcover protocol ensure that all journalists broadcast their cover keys and
start \phrecvproc after receiving a new cover key. Therefore, \name satisfies
the requirements on the communication systems in Theorem~\ref{thm:unobservable}.
As a result, non-participants cannot detect whether users communicate. Thus,
protecting the identity of participants as required.

Third, \name aims to hide the identity of journalists from their counterparts
in a conversation. Theorem~\ref{thm:anonymity} shows that the communication
system does not reveal the identity of journalists at the network layer. \name
also ensures protection at the cryptographic layer: as we argued above, queries
are unlinkable.
However, \name cannot provide unconditional protection for conversations.
Queriers or document owners could identify themselves as part of the
conversation. Moreover, by their very nature, messages in a conversation are
linkable. Also, as we discuss below, insiders can use extra information to
identify communication partners.

\para{Protecting Document Collections}
\label{para:search-privacy}
\emph{Any} functional search system inherently reveals information about the
documents that it makes available for search: To be useful it must return at
least one bit of information. An attacker can learn more information by making
additional queries. We show that \name provides comparable document owner’s
privacy to that of ideal theoretical search systems. We use as a security metric
the number of queries an attacker has to make to achieve each of the following
goals:

\parait{Document Recovery.} 
Given a target set of keywords (e.g. ``XKeyscore'' and
``Snowden''), an adversary aims to learn which of these target keywords are
contained in a document for which some keywords are already known.

\parait{Corpus Extraction.} 
Given a set of target keywords, an adversary aims to learn which documents in
a corpus contain which target keywords. If the target set contains all possible
keywords, the adversary effectively recovers the full corpus.

Any functional search system is also susceptible to confirmation attacks. An
adversary interested in knowing whether a document in a collection contains a
keyword (e.g., ``XKeyscore'' to learn whether the collection contains the
Snowden documents) can always directly query for the keyword of interest.

We compare the number of queries an adversary needs to extract the corpus or
recover a document in the following three settings: when using \name, and when
using one of two hypothetical systems. The first hypothetical system, called
\emph{1-bit}, is an ideal search system. In this system, given a query, the querier
learns \emph{only} one bit of information: whether the owner has a matching
document. The second hypothetical system, called \emph{\#doc}, is an ideal
search system where the querier learns how many matching documents the owner
has.

\begin{table}[tb]
\centering
\caption{Privacy and scalability of the hypothetical and \name's \mspsi based
  search protocols. The table shows the number of queries necessary to achieve
  document recovery and corpus extraction, when interacting with a corpus of $d$
  documents over a set $n$ keywords. The document extraction bound for the 1-bit
  system extracts up to uniqueness bound $\uniqueNum$.}
\label{tab:document-privacy}
\begin{tabular}{lcccc}
  \toprule
                  & Doc                          & Extract                      & Scale        \\
  \midrule
  1-bit           &  $n$                         & $n^u + nd$\tnote{*}          & -\hskip 1pt -\\
  \#doc           &  $n$                         & $nd$                         & -            \\
  \textbf{\name}  & $\mathbf{n/}\textbf{\kwlim}$ & $\mathbf{n/\textbf{\kwlim}}$ & \textbf{+}   \\
  \bottomrule
\end{tabular}
\end{table}

Table~\ref{tab:document-privacy} compares these hypothetical systems with
\name's use of \mspsi, where $d$ is the number of documents and $n$ the number
of relevant keywords. 
We show that extracting all the keywords from a document requires at most $n$
queries in the 1-bit and \#docs search systems \fullversioncmd{in
Appendices~\ref{appendix:sub:one_bit_search} and
~\ref{appendix:sub:numdocsearch_search_extraction}.}
\confversioncmd{in the extended version~\cite{datashareextended}
(Appendices B.1 and B.2).}

Extracting the full corpus using the 1-bit search system is not always
possible. Let the \emph{uniqueness number} $\uniqueNum_D$ be the smallest
number of keywords that uniquely identify a document $D$. If $D$ is a strict
subset of another document $D'$, the document cannot be uniquely identified, and
we set $\uniqueNum_D = \infty$. However, as corpora are small, we expect that
most documents can be identified by a few well-chosen keywords, resulting in
small uniqueness numbers.

\fullversioncmd{In Appendix~\ref{appendix:sub:one_bit_search},}
\confversioncmd{In Appendix~B.1 of the extended version~\cite{datashareextended},}
we show that extracting all documents with uniqueness number less or equal to
$\uniqueNum$ takes $O(n^\uniqueNum + nd)$ queries in the 1-bit search system. 
\fullversioncmd{In Appendix~\ref{appendix:sub:numdocsearch_search_extraction}}
\confversioncmd{In Appendix~B.2 of the extended version~\cite{datashareextended},}
we show that extracting all documents (regardless of uniqueness number) takes
$O(nd)$ queries in the \#doc search system.

In \name, we limit \mspsi queries to $\kwlim$ keywords per query. Hence, any
document extraction attack must make at least $n / \kwlim$ queries to ensure all
keywords are queried at least once. In fact, this bound is tight for both
document recovery and corpus extraction for \mspsi: By making $n / \kwlim$
queries with $\kwlim$ keywords each, the attacker learns which keywords are
contained in which documents.

In summary, \name offers similar protection against corpus extraction as the
\numdocsearch ideal system. For document recovery, not even the ideal
1-bit-search system offers much better protection. At the same time, \mspsi is
much more efficient than their ideal counterparts.\\

\para{Internal Adversaries}
\label{para:aux_privacy}
We now discuss how an adversary may use auxiliary information about a
journalist's behavior or corpus to gain an advantage in identifying the
journalist. Some of these attacks are inherent to all systems that provide
search or messaging capabilities. These attacks, however, \emph{do not permit}
the adversary to extract additional information from journalists' corpora. 

\parait{Intersection Attacks.} 
A malicious sender (respectively, receiver) who has access to the online/offline
status of journalists can use this information to reduce the anonymity set of
the receiver (respectively, sender) to only those users that are online. As more
messages are exchanged, this anonymity set becomes unavoidably
smaller~\cite{KedoganAP02}. This attack is inherent to all low-delay
asynchronous messaging systems, including the one provided by the communication
server. In the context of \name, we note that once document owners and queriers
are having a conversation, it is likely that they reveal their identity to each
other. Yet, we stress that preserving anonymity and, in general, that minimizing
the digital traces left by the journalists in the system is very important to
reducing the risk that journalists become profitable targets for subpoenas or
hacking attempts. 

\parait{Stylometry.}
A malicious receiver can use stylometry, i.e., linguistic style, to guess the
identity of the sender of a message. The effectiveness of this attack depends on
the volume of conversation~\cite{lopez13,muthuselvi16}. This attack is inherent
to all messaging systems, as revealing the content of the messages is required
to provide utility.

\parait{Partial Knowledge of Corpus.}
Adversaries who have prior knowledge about a journalist's corpus can use this
knowledge to identify this journalist in the system. However, due to \mspsi's
privacy property (see Theorem~\ref{thm:mspsi-privacy}), learning \emph{more}
about the documents in this journalist's corpus requires making search queries.

In particular, if an adversary convinces a journalist to add a document with a
unique keyword pattern to his corpus, then the adversary can detect this
journalist's corpus by searching for the pattern. \name cannot prevent such
out-of-band watermarking. However, the adversary still needs to make further
queries to learn anything about non-watermarked documents in the collection.

\para{Non-goals} 
Finally, we discuss security properties that are not required in \name.

\parait{Query Unlinkability.} 
\label{par:query_unlinkability}
\name does not necessarily hide which queries are made by the same querier. Even
though anonymity is ensured at the network and application layers, queriers that
have made multiple queries may retrieve responses for all these queries in quick
succession after coming online. Document owners know the corresponding query of
their messages, and if they collude with the \DSpigeonhole, then they can infer
that the same person made these queries. As no adversary can learn any
information about the queries themselves, we consider this leakage to be
irrelevant.

\newcommand{\ringkw}{\textvar{link}}
\parait{Owner Unlinkability.} 
\label{par:response_unlinkability}
\name also reveals \emph{which} pseudonymous document owner created a \mspsi
response, making responses linkable. \name cannot provide unlinkability for
document owners when using \mspsi. Although \mspsi itself could be modified to
work without knowing the document owner's pseudonym, an adversary could simply
repeat a specific rare keyword (for example, ``one-word-to-link-them-all'') and
identify the document owners based on the corresponding pretag that they produce
for the rare keyword. We believe that revealing the document owner's pseudonym
is an acceptable leakage for the performance gain it provides.

%% file: parts/s6-evaluation.tex
\subsection{Cost Evaluation} 
\label{sec:evaluation}

\begin{figure*}[tbp]
  \centering
  \includegraphics[width=0.32\textwidth]{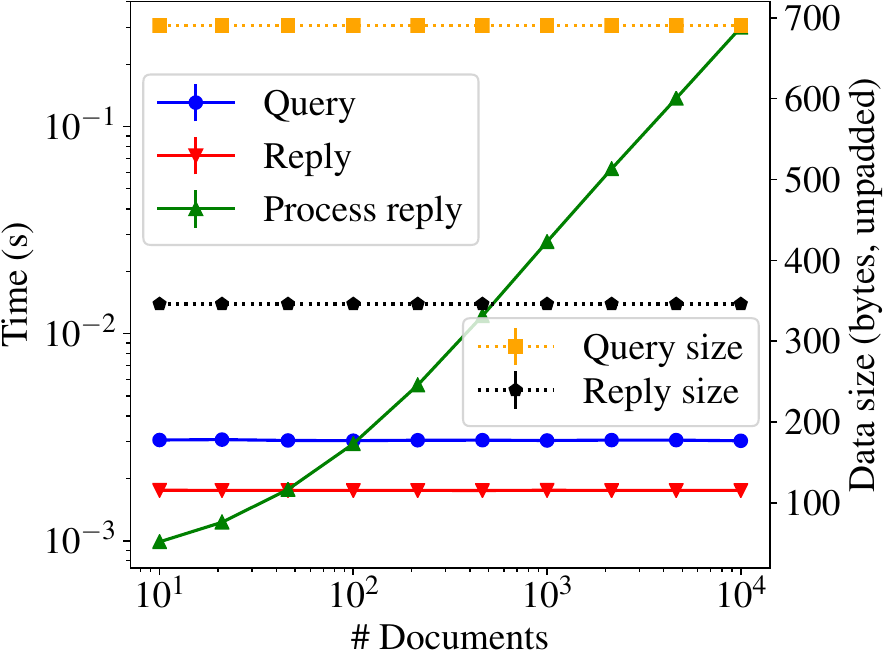}
  \hfill
  \includegraphics[width=0.32\textwidth]{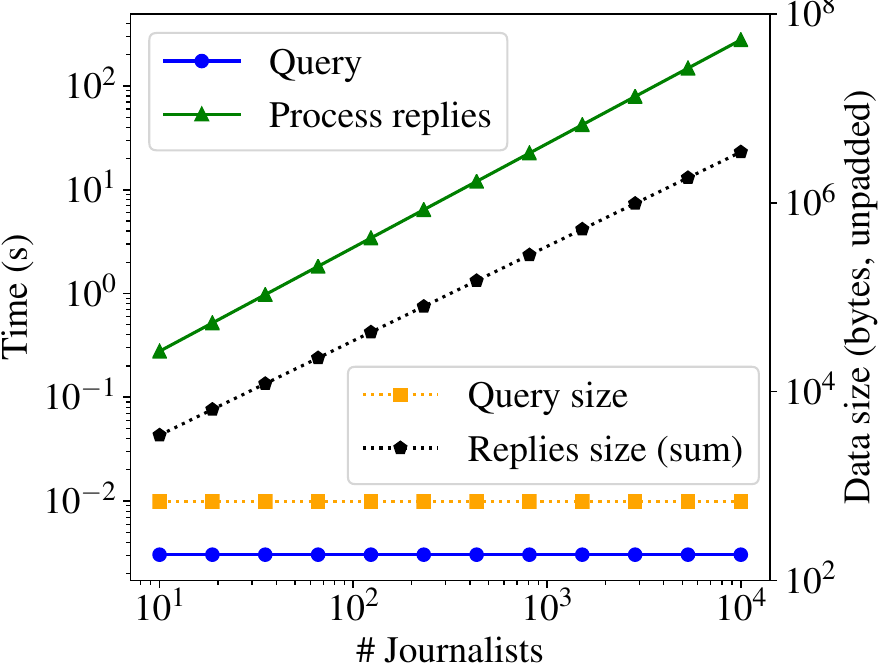}
  \hfill
  \includegraphics[width=0.32\textwidth]{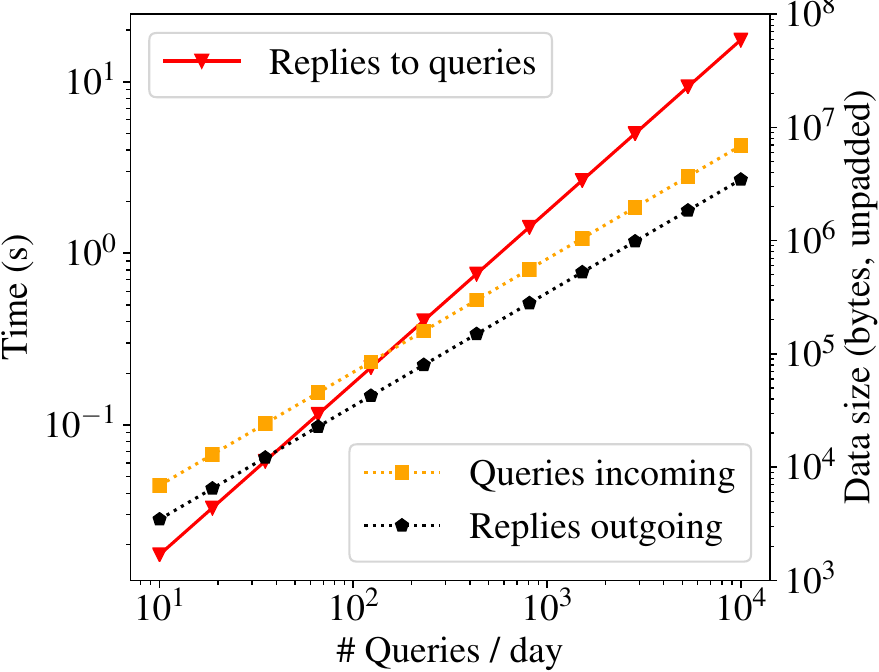}
  \caption{Time (left axis) and bandwidth (right axis, unpadded) for single
   query on one journalist (left), single query on all journalists (center),
   answering several queries (right).}
  \label{fig:perf:all}
\end{figure*}

At the time of writing, ICIJ has implemented the local
search and indexing component of \name~\cite{DatasharePlatform}. In addition, we have 
implemented a Python prototype of the cryptographic building blocks underlying
search (Section~\ref{sec:multi_set_psi}) and authentication
(Section~\ref{sub:preliminaries}).\footnote{The code is open source and
  available at:
\url{https://github.com/spring-epfl/datashare-network-crypto}} We did not
implement the messaging service (Section~\ref{sec:messaging}), as it relies on
standard building blocks and cryptographic operations.

To agree on the final configuration of the system, we are currently running a
user study among the organization members. The goal is to familiarize
journalists with a type of search and messaging system that is different than
those they typically use in their daily activities (Google and email or instant
messaging, respectively), as well as with the threat model within which \name
provides protection. We recall that \name hides all key management and
cryptography from the users, hence we do not study those aspects.

In this section, we evaluate the performance of the cryptographic operations
involved in search and authentication. Our prototype uses the
\textsf{petlib}~\cite{Petlib} binding to OpenSSL on the fast NIST P-256 curve
for the elliptic curve cryptography in \mspsi. We implement the Cuckoo filter
using \textsf{cuckoopy}~\cite{Cuckoopy}. We ran all experiments on an Intel
i3-8100 processor running at 3.60GHz using a single core. We note that
operations could be easily parallelized to improve performance.

We focus our evaluation on the computational cost and bandwidth cost of the
authentication and search primitives to ensure that \name fulfills the
requirements in Section~\ref{sec:system:requirements} without journalists
needing fast hardware or fast connections. When reporting bandwidth cost, we omit
the overhead of the meta-protocol that carries messages between system parties.
We do not consider any one-time setup cost or the standard cryptography used
for messaging. We also do not measure network delay as the latency the Tor
network introduces -- around one second~\cite{TorPerf} -- is negligible compared
to the waiting time imposed by connection asynchrony; and it is orders of
magnitude less than the journalists waiting limits (see
Section~\ref{sec:system:requirements}). 

We provide performance measurements for different system work loads. We consider
the \emph{base scenario} to be 1000 journalists, each of whom makes 1000
documents available for search. There is no requirement for the number of
keywords per document or keywords per query. For a \emph{conservative} estimate,
we assume that each document contains 100 keywords, and that each query contains
10 keywords.

\para{Authenticating Journalists} 
We implement the \bs scheme using Abe's blind signatures~\cite{Abe01}. Running
\bssign requires transferring 413 bytes and takes 0.32\,ms and 0.62\,ms,
respectively, for the organization and the journalist. Each blind signature is
360 bytes, and verifying it using \bsverify takes 0.4\,ms. We include these
costs in the respective protocols.

\para{Publishing Documents} 
Data owners run \publish to make their documents searchable. For the base
scenario, this \emph{one-time} operation takes 14\,seconds and results in a
cuckoo filter of size 400\,KB for a FPR of 0.004\%. For a conservative
estimation, we assume all keywords are different. When documents contain
duplicate elements $y$, the precomputation can be amortized: the pretag
$\GHash(y)^{\serverkey}$ has to be computed only once. 

\para{Querying a Single Journalist}
Figure~\ref{fig:perf:all}, left, shows the time and bandwidth required to issue
one query on one collection, depending on the collection size. The querier
constructs the query using \querycmd and sends it to the document owner (the
querier's computation cost includes the cost of obtaining the one-time-use token
using \gettoken). The document owner responds using \reply. These operations are
independent of the number of documents. The querier runs $\process$ to retrieve
the responses, and to compute the intersection of query and collection. This
takes 27\,ms in the base scenario. Bandwidth cost reflects the raw content size.
But recall that, in practice, the messaging system pads messages to 1\,KB.

\para{Querying All Journalists} 
As expected, the processing time and bandwidth of \querycmd are independent of
the population size, whereas the cost of processing the responses grows linearly
with the number of queried journalists (Figure~\ref{fig:perf:all}, center). For
the baseline scenario, processing all 999 responses takes about 27 seconds
\emph{in total} and requires retrieving 1\,MB of padded responses. We note that
this cost is only paid by the querier, and does not impact the document owners
(see below). Moreover, as replies are unlikely to arrive all at once,
processing can be spread out over time; thus reducing the burden on the querier's
machine.

This computation assumes that each journalist has the same number of documents.
In practice, this might not hold. However, as we see in Figure~\ref{fig:perf:all},
left, as soon as collections have more than 50 documents the computation time
grows linearly with the collection size. Hence, as long as journalists have
collections with at least 50 documents, the measurements in
Figure~\ref{fig:perf:all}, center, are largely independent of how these
documents are distributed among journalists.

\begin{figure}[tbp]
  \centering
  \includegraphics[width=0.7\columnwidth]{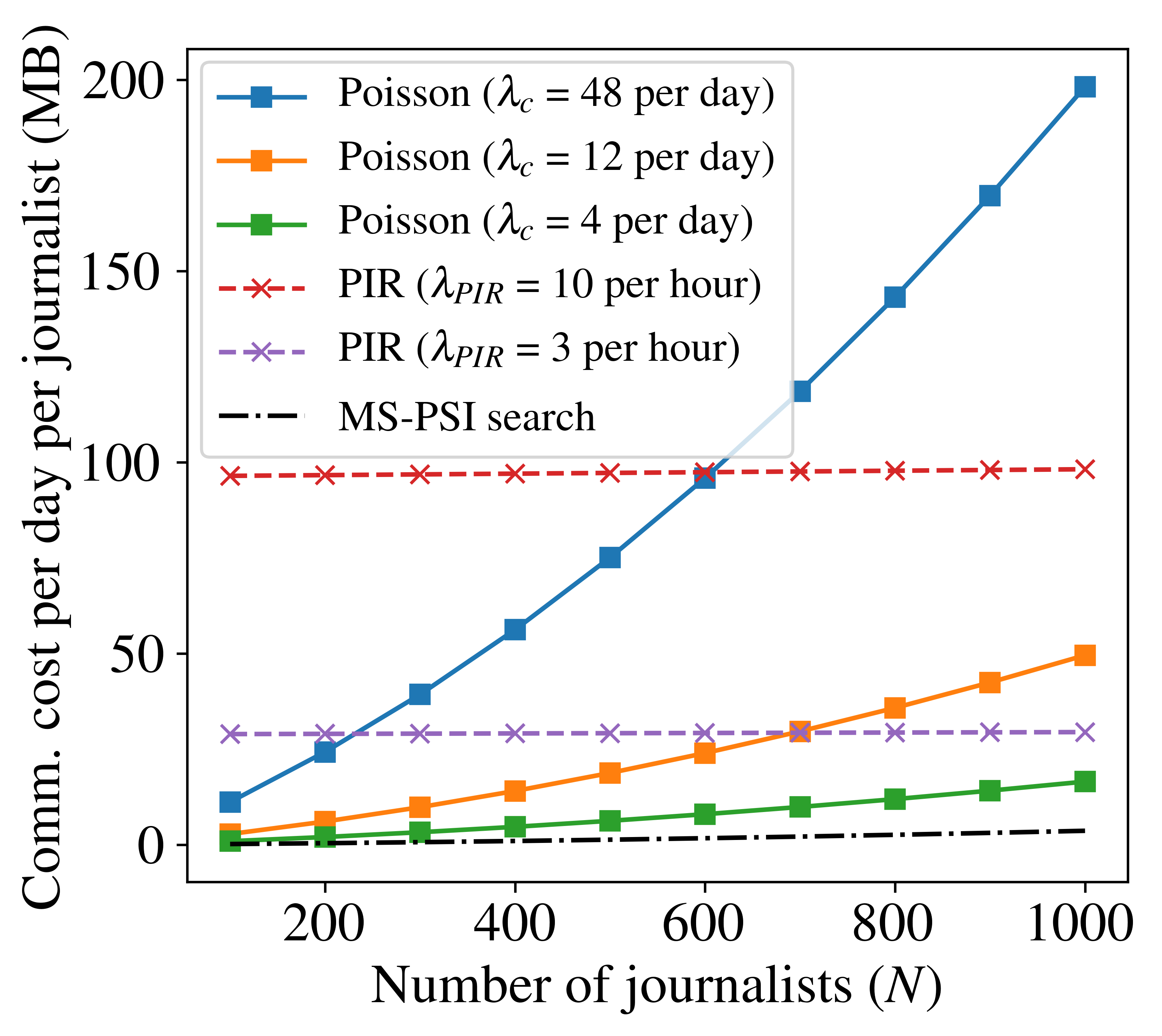} 
  \caption{Communication cost for different communication
  strategies, depending on the number of journalists. We assume 1 query per
  journalist per day in the search component.}
  \label{fig:overall}
\end{figure}

\para{The Cost for Document Owners}
Document owners spend time and bandwidth to answer queries from other
journalists. Figure~\ref{fig:perf:all}, right, shows how these costs depend on
the total number of queries an owner receives per day. Even when all journalists
make 10 queries of 10 keywords each day (unlikely in practice) the \emph{total}
computation time for document owners is less than 20 seconds; and they send and
receive less than 7 megabytes (10\,MB when padded).

\para{Overall Cost of \name}
Finally, we plot in Figure~\ref{fig:overall} the \emph{total} bandwidth a
journalist needs per day to run \name, depending on the number of journalists in
the system and the strategy implemented by the communication system. Regardless
of the size of the system, the cost associated to hide communications dominates
the cost stemming from searches. Regarding the communication cost, as explained
in Section~\ref{sec:mess_eval}, for small organizations Poisson-rate cover
traffic provides a better trade-off with respect to throughput, but as more
journalists join the system, the PIR-based system starts performing better.

%% file: parts/s7-related-work.tex
\section{Related Work} 
\label{sec:related_work}

Many PSI protocols~\cite{CristofaroT10, NaorR04, HazayL10,KissL0AP17} differ
from that of De Cristofaro et al.~\cite{CristofaroGT12}, but only in how they
instantiate the oblivious pseudorandom functions (OPRFs). Our \mspsi protocols
can easily be adjusted to use alternative OPRFs to compute the pretags. As 
bandwidth is at a premium in our scenario, we base our \mspsi protocols
on the scheme of De Cristofaro et al. as it has the lowest
communication cost.

The restrictions on computational power and bandwidth rule out many other PSI
schemes. Protocols based on oblivious polynomial evaluation~\cite{FreedmanNP04}
have very high computational cost. Hash-based PSI
protocols~\cite{Pinkas0Z14,PinkasSZ18, Pinkas0WW18} have low computational cost,
but require much communication. Finally, PSI protocols can be built from generic
secure multi-party computation
directly~\cite{HuangEK12,Pinkas0WW18,Pinkas0SZ15}. However, this approach also
suffers from a high communication cost and requires more than one communication
round.

Secure multi-party computation based PSI protocols can be extended to provide
better privacy than \mspsi: The underlying circuits can be extended to implement
either the ideal 1-bit search or the \numdocsearch search system. However, their
high communication and round complexity rule out their use in our document
search system. Recently, Zhao and Chow proposed a threshold PSI protocol based
on polynomial evaluation~\cite{Zhao1C18} that can implement the \numdocsearch
search system (by setting the threshold equal to the number of keywords). But
its communication and computation complexity rule it out.

A document search engine could also be implemented using private information
retrieval (PIR): Queriers use PIR to privately query keywords in the document
owner's database. Computational PIR
protocols\cite{KushilevitzO97,AngelCLS18,MelchorBFK16} (IT-PIR
protocols~\cite{BeimelI01,Goldberg07} do not apply) place a high computational
burden on the database owner. More importantly, PIR requires a fixed set of keywords,
that cannot exist for the journalists' use case. Keyword-based PIR
approaches~\cite{AngelS16,ChorGN97} sidestep this issue, but instead require
multiple communication rounds. Therefore, PIR cannot be used in our scenario.

Encrypted databases hide the queries of data owners from an untrusted database
server~\cite{SongWP00,PopaRZB11,PappasKVKMCGKB14,EtemadKPE18}. Although
\DataShare could operate such a central encrypted database, this system would
not be secure. On the one hand, if the encrypted database is used as a central
service for all collections, then a collusion between a journalist and the
database server would leak the entire database. This would violate document
privacy. On the other hand, if each journalist operates a personal database,
then collusion between the database server and the document owner (acting as the
`data owner' in the terminology used in the encrypted database literature) might
leak search queries, as these systems are not designed to hide queries from a
database server that colludes with the data owner. This would violate query
privacy.

%% file: parts/s8-conclusions.tex
\section{Future Steps: Better Protection} 
\label{sec:conclusions}
We have introduced \name, a decentralized privacy-preserving search engine that
enables journalists to find and request information held by their peers. \name
has great potential to help journalists collaborate in uncovering cross-border
crimes, corruption, or abuse of power.

Our collaboration with a large organization of investigative journalists (ICIJ)
provided us with a novel set of requirements that, despite being deeply grounded
in practicality, are rarely considered in academic publications. These
requirements led us to design new building blocks that we optimized for
security trade-offs different than previous work. We combined these building
blocks into an efficient and low-risk decentralized search system.  

Yet, \name's protections are not perfect. Both the search primitive, and the
availability of timestamps of actions in the system, leak information. At the
time of writing, the high cost in bandwidth and/or computation of
state-of-the-art techniques that could prevent this leakage -- e.g.,  PIR to
hide access patterns and efficient garbled circuits to implement one-bit search
-- precludes their deployment.

We hope that this paper fosters new research that addresses these problems. We
believe that the new set of requirements opens an interesting new design space
with much potential to produce results that have a high impact, not only by
helping investigative journalism to support democratic societies, but also in
other domains.

%% file: parts/s11-proof.tex
\newcommand{\view}{\textvar{View}}
\newcommand{\ideal}{\textvar{Ideal}}
\newcommand{\real}{\textvar{Real}}
\newcommand{\prooforacle}[2]{\mathcal{O}_{#1}^{#2}}
\renewcommand{\adv}{\mathcal{A}}
\renewcommand{\sim}{\mathcal{S}}
\newcommand{\advb}{\mathcal{B}}
\newcommand{\tcsize}{\mathcal{N}}
\newcommand{\psiadaptive}{\textvar{PSI}_{\textvar{adt}}}

\section{Security of \mspsi}
\label{sec:formal_proof}
In this section, we prove that \mspsi is correct and private. Proving privacy
requires showing that neither a malicious client nor a malicious server can
learn anything beyond the intended output of the protocol. The client's
interaction with the server is identical to the \psi ~\cite{CristofaroGT12} and
\cspsi ~\cite{KissL0AP17} protocols. Hence, they have the same client privacy
against a malicious server. To prove server privacy, we use the ideal/real world
paradigm in the random oracle model and show that a malicious \mspsi client does
not learn anything beyond the intended output of the protocol as long as the
One-more-Gap-DH assumption holds. We first prove correctness.

\begin{theorem}
  \label{thm:mspsi-correctness}
  The \mspsi protocol is correct.
\end{theorem}
\begin{proof}  We show that the intersection $I_d$ of the $d$'th set represented as $Y_d =
  \{y_{d,1}, y_{d,2}, .., y_{d,n_d}\}$ and the client's set $X = \{x_1, x_2, ..,
  x_m\}$ is equal to $I_d = X \cap Y_d$.

  Recall that the client computes the intersection with set $Y_d$ as \(I_d =
  \{x_i \mid \clientTags^{(d)}_i \in \tagCollection \}\). We prove that \(
  \clientTags^{(d)}_i \in \tagCollection\) iff \(x_i \in X \cap Y_d\).
  For each client keyword $x_i$, the client computes the pretag
  \( 
    \clientPreTag_i = \xpp_i^{\clientkey^{-1}} =
    \xp_i^{\serverkey \clientkey^{-1}} = \GHash(\x_i)^{\clientkey \serverkey \clientkey^{-1}} =
    \GHash(\x_i)^{\serverkey} 
  \). 
  On the other hand, the server computes its pretags for set $Y_i$ as 
  \(\pretagset{i} = \{\GHash(\y)^{\serverkey} \mid y \in Y_i\}\)
  and computes its tag collection as 
  \(
     \tagCollection = \left\{\hash(i \concat t) \,\middle|\,  i \in \NatNumUpTo{\serversize} \land t \in \pretagset{i}\right\}
                    = \left\{\hash\left(i \concat \GHash(\y)^{\serverkey} \right)  \,\middle|\, 
                              i \in \NatNumUpTo{\serversize} \land y \in Y_i\right\}
  \).
  Hence, the intersection will be computed as
  \begin{align*}
    I_d &= \{ x_i \mid \clientTags^{(d)}_i \in \tagCollection \} \\
        &= \bigg\{x_i \,\bigg|\, \hash(d \concat \clientPreTag_i) \in
              \Big\{\hash\big(i \concat \GHash(\y)^{\serverkey} \big) \,\Big|\, 
              i \in \NatNumUpTo{\serversize} \land y \in Y_i\Big\} \bigg\}\text{.}
  \end{align*}
  The hash functions $\GHash$ and $\hash$ are cryptographically secure, and
  the probability of collision is negligible. Hence, two hash values will only be
  equal when their inputs are equal. Since $d$ is an input to $\hash$, only the
  keywords from the $d$'th set in the server's tag collection can be in the
  intersection. Therefore:
  \[I_d = \left\{x_i \,\middle|\, \GHash(\x_i)^{\serverkey} \in
            \{\GHash(\y)^{\serverkey} \mid y \in Y_d\} \right\} \text{.}\]
            
 Similarly, $x_i$ is an input to $\GHash$ and it will be in the
 intersection $I_d$ if $x_i$ is present in both $X$ and $Y_d$ sets. Consequently:
  \[I_d = \{x_i \mid \x_i \in \{\y \mid y \in Y_d\} \} 
        = \{x_i \mid \x_i \in Y_d \} = X \cap Y_d \text{.} \qedhere \]
\end{proof}

The \mspsi protocol is interactive: the client asks keywords in multiple queries
and receives the response of the $i$'th query before making the $i+1$'th one. To
measure the client's interaction with the server, we define $q$ as the number of
queried keywords. We chose the number of queried keywords over the number of
queries since the server reveals the same information about these $q$ keywords
regardless of how many queries they were asked in. Without loss of generality,
we assume an adaptive adversary in which the adversary asks her keywords one by
one and receives responses immediately. The non-adaptive versions or versions
where the client queries multiple keywords simultaneously only delay when the
adversary receives the response. Hence, they have the same security guarantee as
the adaptive version.

\parab{An adaptive PSI functionality.}
Let $\lambda$ be an empty string, $w$ be the client's input keyword, and
\(\mathcal{Y} = [Y_1, \ldots, Y_n]\) be a list of $n$ server sets \(Y_i =
\{y_{i,1}, \ldots, y_{i, n_i}\}\). We define the adaptive \psi functionality
$\psiadaptive$ as a two party function in which the client learns the sets which
contain the keyword $w$, and the server learns nothing:
\[\psiadaptive(w, \mathcal{Y}) = \big(\{i \mid i \in \NatNumUpTo{n} \land w \in
Y_i)\}, \lambda\big)\text{.}\]

We define $\ideal_q$ as an ideal instantiation of $\psiadaptive$ in which a
trusted third party receives the server's input and responds to the client's
$\psiadaptive$ queries at most $q$ times. $\ideal_q$ provides an oracle
\(\prooforacle{\ideal}{}(\mathcal{Y}, w) \rightarrow \{i \mid i \in
\NatNumUpTo{n} \land w \in Y_i)\} \)  which responds to ideal queries.
The ideal oracle can answer non-adaptive queries with $t$ keywords by calling
the $\psiadaptive$ process $t$ times and concatenating their responses. Note
that this operation costs $t$ adaptive queries.

We define $\real_q$ as the real world instantiation of $\psiadaptive$ which runs
the \mspsi protocol and allows the client to ask up to $q$ keywords.
The \mspsi protocol consists of 2 parts: 1) publish, which corresponds to the
pre-process phase and 2) exponentiation, which corresponds to the online
interaction.
The simulation implicitly assumes a known fixed size for the parties' inputs as
an adversary can distinguish different input sizes. To make this explicit, we
reveal the number of server sets $N$ and the size of the tag collection $\tcsize
= |\tagCollection| = \sum_{i=1}^{N}n_i$ to the simulator and the adversary. Bear
in mind that the \mspsi protocol reveals an upper bound on $N$ and $\tcsize$. 
\mspsi uses two hash-functions $\hash: \bin^* \rightarrow \bin^l$ and $\GHash:
\bin^* \rightarrow \G$, which are modeled in the random oracle model (ROM) as
oracles $\prooforacle{H_{kw}}{}$ and $\prooforacle{H_{G}}{}$ respectively. We
define the following oracles to represent $\real_q$:

\begin{description}
  \item[$x \leftarrow \prooforacle{H_{G}}{\real}(w)$] hashes the keyword $w \in
    \bin^*$ into a uniformly random group element $x \in_{R} \G$.
    \item[$TC \leftarrow \prooforacle{Pub}{\real}()$] pre-processes the server's
    input, i.e., chooses the server's secret key $\alpha$, and publishes the
    server's tag collection 
    \(\tagCollection = \{\hash(i \concat \GHash (y)^\alpha ) \mid 
    i \in \NatNumUpTo{\serversize} \land y \in Y_{i}\}\).
    \item[$x^{\alpha} \leftarrow \prooforacle{exp}{\real}(x)$] takes a group
    element $x \in \G$ and returns $x^\alpha$. The
    adversary is limited to making up to $q$ queries.
    \item[$\tau \leftarrow \prooforacle{H_{kw}}{\real}(\omega)$] hashes the
      input $\omega \in \bin^*$ to a random $l$-bit tag $\tau \in \bin^l$.
\end{description}

To show that $\ideal_q$ and $\real_q$ have the same server privacy guarantee, we
assume a PPT adversary $\adv$ that interacts with $\real_q$ and design a
simulator $\sim$ which given black-box access to $\adv$ extracts the same
information from the ideal world $\ideal_q$. In Theorem~\ref{thm:secrecy}, we
proved that the \mspsi is correct. Since \mspsi is correct and its output is
deterministic, we only have to prove the following computational
indistinguishability to show that the server privacy in the real and ideal
worlds are equivalent:
\[\view_{\real_q}^{\adv}\big([w_i]_{i \in \NatNumUpTo{q}}, \mathcal{Y}\big)
\overset{c}{\equiv} 
\view_{\ideal_q}^{\sim^{\adv, \prooforacle{\ideal}{}(\mathcal{Y}, \cdot)}}
\big([w_i]_{i \in \NatNumUpTo{q}}, \mathcal{Y}\big)\text{.}\]
Where $\view_{f}^{P}(X, Y)$ is the view of $P$ in an execution of $\psiadaptive$
instantiated with $f$, $X$ is the client's input, and $Y$ is the server's input.

We start with a high-level overview of the proof. We build simulator $\sim$ by
constructing oracles to represent \mspsi. Afterward, we show that the adversary
cannot distinguish simulator $\sim$ from $\real_q$.
Oracles $\prooforacle{H_G}{\sim}$ and $\prooforacle{exp}{\sim}$ are similar to
their real world counterparts. The oracle $\prooforacle{Pub}{\real}$ follows the
pre-process phase of the \mspsi protocol to compute the server's tag collection
$\tagCollection$ and produces a set of $\tcsize$ random $l$-bit tags generated
by the hash function $\hash$. To mimic this, the oracle
$\prooforacle{Pub}{\sim}$ returns $\tcsize$ random $l$-bit tags. To construct
the oracle $\prooforacle{H_{kw}}{\sim}$, simulator $\sim$ has to extract the
adversary's effective input and query it to the ideal oracle
$\prooforacle{\ideal}{}$ to respond accordingly (i.e. with one of the $\tcsize$
random outputs of $\prooforacle{Pub}{\sim}$ for positive and a uniformly random
tag for a negative response). The key idea in building this oracle is that the
simulator uses the server's secret $\alpha$ to decrypt queries $\omega = d ||
x^\alpha$ to the oracle $\prooforacle{H_{kw}}{\sim}$ and extract the input $x$.
After extracting $\adv$'s input, $\sim$ queries the ideal oracle and responds
accordingly. If $\sim$ makes more than $q$ queries from
$\prooforacle{\ideal}{}$, the simulation fails, and the simulator aborts. 

First, we prove in Lemma~\ref{lm:ind-real-sim} that as long as the simulator
does not abort, $\sim$ is indistinguishable from $\real_q$. Second, we show in
Lemma~\ref{lm:neg-abort} that the probability of abort is negligible. The
simulator aborts when an adversary queries ${q+1}$ distinct keywords from
$\prooforacle{H_{kw}}{\sim}$ without querying $\prooforacle{exp}{\sim}$ more
than $q$ times. Informally, this means that the adversary can compute the
exponentiation $r^\alpha_i$ of $q+1$ random values $r_i$ with only $q$ queries
to the exponentiation oracle, which translates into the One-more-Gap-DH problem.
We show that if an adversary $\adv$ exists which has a non-negligible chance of
forcing an abort, we can build an adversary $\advb$ that can break the
One-more-Gap-DH assumption given black-box access to $\adv$.

We now give details on how to build simulator $\sim$:
\begin{description}[topsep=4pt]
  \item[$x \leftarrow \prooforacle{H_{G}}{\sim}(w)$] responds the same way as
    $\prooforacle{H_{G}}{\real}$; stores the mapping between each keyword and
    its matching element.
  \item[$TC \leftarrow \prooforacle{Pub}{\sim}()$] choses a random key $\alpha$
    and generates the tag collection $\tagCollection$ as a set of $\tcsize$
    uniformly random $l$-bit tags. Recall that the simulator receives $N$ and
    $\tcsize$ as input.
  \item[$x^{\alpha} \leftarrow \prooforacle{exp}{\sim}(x)$] same as
    $\prooforacle{exp}{\real}$. The adversary is limited to making up to $q$
    queries.
  \item[$\tau \leftarrow \prooforacle{H_{kw}}{\sim}(\omega)$] the oracle
    responds to repeated queries consistently. For a new query $\omega$, it
    proceeds as follows:
    \begin{enumerate}
      \item Parse the input $\omega$ as ``$d || z$'' where $d \in \Zn$ and $z
        \in \G$. If this fails, respond with a random $l$-bit tag $\tau$.
      \item Use the secret key $\alpha$ to compute the adversary's effective
        input element $x = z^{\alpha^{-1}}$.
      \item If $x$ is the result of a query to the $\prooforacle{H_{G}}{\sim}$,
        let $w$ be the corresponding preimage. Otherwise, return a random
        $l$-bit tag $\tau$.
      \item If $w$ has not been queried, query $w$ from the ideal oracle
        $\prooforacle{\ideal}{}(\mathcal{Y}, \cdot)$ and store the response.
      \item If simulator $\sim$ has queried the ideal oracle more than $q$
        times, then abort.
      \item Respond with a random unused tag $\tau \in TC$ if $d \in
        \prooforacle{\ideal}{}(\mathcal{Y}, w)$. Otherwise, respond with a
        random $l$-bit tag $\tau$.
    \end{enumerate} 
\end{description}

\begin{lemma}
  \label{lm:ind-real-sim}
  The simulator $\sim$ is indistinguishable from the $\real_q$ as long as
  $\sim$ does not abort.
\end{lemma}

\begin{proof}
  Oracles $\prooforacle{exp}{\real}$ and $\prooforacle{exp}{\sim}$ are
  identical, and it is easy to see that $\prooforacle{H_{G}}{\sim}$ and
  $\prooforacle{Pub}{\sim}$ are indistinguishable from their $\real$
  counterparts as their output is uniformly random.

  The \mspsi protocol uses $\prooforacle{H_{kw}}{\real}$ to produce $l$-bit
  tags. Theorem~\ref{thm:mspsi-correctness} proves the correctness of the \mspsi
  protocol and shows that a final tag $\clientTags^{(d)} \leftarrow
  \prooforacle{H_{kw}}{\real}(\omega)$ is in the server's tag collection
  $\tagCollection$ if \(\omega = \text{``}d \concat \GHash(w)^s\text{''}\) and
  $w \in Y_d$ where $s$ is the server's secret key. Otherwise,
  $\clientTags^{(d)}$ is a random tag.
  Similarly, the oracle $\prooforacle{H_{kw}}{\sim}(\omega)$ responds with a tag
  \(\tau \in \tagCollection \) if \(\omega =  \text{``}d \concat
  \prooforacle{H_G}{\sim}(w)^\alpha  \text{''}\) and \(d \in
  \prooforacle{\ideal}{}(\mathcal{Y}, w)\) where $\alpha$ is the simulator's
  secret key. Otherwise, the oracle responds with a random $l$-bit tag. As long
  as the oracle $\prooforacle{H_{kw}}{\sim}$ correctly detects the effective
  input $(w, d)$ and its status $w \in Y_d$, the adversary cannot distinguish
  the oracles $\prooforacle{H_{kw}}{\sim}$ and $\prooforacle{H_{kw}}{\real}$.
  There are two possible cases for an incorrect response: false positives and
  false negatives. Now, we show that the probability of incorrect response is
  negligible.
  Let $q_G$ and $q_{kw}$ be the number of queries to oracles
  $\prooforacle{H_{G}}{\sim}$ and $\prooforacle{H_{kw}}{\sim}$ respectively. A
  false positive happens when there is a collision between the adversary's input
  $z = x^\alpha$ to the oracle $\prooforacle{H_{kw}}{\sim}$ and an unintended
  keyword queried from $x \leftarrow \prooforacle{H_{G}}{\sim}(w)$. This event
  has a probability of $q_G \cdot q_{kw} / \text{Ord}(\G)$ due to the randomness
  of $\prooforacle{H_{G}}{\sim}$. A false negative happens when $x \leftarrow
  \prooforacle{H_{G}}{\sim}(w)$ is not known at the time of the
  $\prooforacle{H_{kw}}{\sim}$ query. The probability of
  $\prooforacle{H_{G}}{\sim}(w)$ matching one of previous
  $\prooforacle{H_{kw}}{\sim}(\omega)$ queries is $q_{kw} / \text{Ord}(\G)$,
  which limits the false negative probability to $q_G \cdot q_{kw} /
  \text{Ord}(\G)$.
\end{proof}

The simulator $\sim$ aborts when adversary $\adv$ queries $q+1$ distinct
keywords from $\prooforacle{H_{kw}}{\sim}$ while making at most $q$ queries from
$\prooforacle{exp}{\sim}$. We assume that $\adv$ triggers an abort with the
probability of $\epsilon$. We state the One-more-Gap-DH assumption and relate it
to $\epsilon$.

\parab{The One-more-Gap-DH Assumption} informally states that computing CDH is
hard even if the adversary has access to a CDH oracle and the DDH problem is
easy.

The adversary $\adv$ in the One-more-Gap-DH assumption gets access to a CDH
oracle $x^\alpha \leftarrow \prooforacle{CDH}{}(x)$ with the secret $\alpha$ and
a $DL_\alpha$ oracle $1/0 \leftarrow \prooforacle{DL_{\alpha}}{}(x, z)$ which
determines whether a pair of elements $x, z \in \G$ has a discrete logarithm
equal to the oracle's secret, i.e., $\alpha = \log_x(z)$. The $DL_{\alpha}$
oracle is a weaker form of the DDH oracle since $\prooforacle{DL_{\alpha}}{}(x,
z) = DDH(h, h^\alpha, x, z)$.

The One-more-Gap-DH assumption states that the adversary has negligible chance
in producing $q+1$ DH pairs $(x_i, x_i^\alpha)$ given $M \gg q$ random challenge
elements $Ch = (c_{1}, .., c_{M}) \in \G^M$ while making at most $q$ queries
to the CDH oracle $\prooforacle{CDH}{}$
\[Pr\left[ \left\{(x_{i}, x_{i}^\alpha) \mid x_i \in Ch\right\}_{i \in
\NatNumUpTo{q+1}} \leftarrow \adv^{\prooforacle{CDH}{}(\cdot),
\prooforacle{DL_{\alpha}}{}(\cdot,\cdot)}(Ch)\right] < \mu\text{.}\]

\begin{lemma}
  \label{lm:neg-abort}
  If the adversary $\adv$ has a non-negligible probability $\epsilon$ in forcing an
  abort in the simulator $\sim$, there exists an adversary $\advb$ which has a
  non-negligible advantage in solving the One-more-Gap-DH problem given black-box
  access to $\adv$.
\end{lemma}

\begin{proof}
  We start with a sketch of the proof. We construct an adversary $\advb$ that
  simulates $\sim$ to adversary $\adv$ to solve the One-more-Gap-DH challenge.
  Simulator $\sim$ has two main functions: computing exponentiations with a
  secret key in $\prooforacle{exp}{\sim}$ and finding the matching input element
  $x = z^{\alpha^{-1}}$ used to query $\prooforacle{H_{kw}}{\sim}$. Adversary
  $\advb$ programs $\prooforacle{H_{G}}{\advb}$ to fix input elements to
  challenge points and uses the CDH oracle $\prooforacle{CDH}{}$ to respond to
  $\prooforacle{exp}{\advb}$ queries. Finally, $\advb$ uses
  $\prooforacle{DL_{\alpha}}{}$ to detect which challenge point matches the
  group element $z$ in the $\prooforacle{H_{kw}}{\advb}$ query $\omega=
  \text{``} d || z\text{''}$. If $\advb$ receives $q+1$ queries corresponding to
  distinct keywords $\{w_i\}_{i \in \NatNumUpTo{q+1}}$ in the oracle
  $\prooforacle{H_{kw}}{\advb}$, then $\advb$ can produce $q+1$ DH pairs from
  the challenge set without querying $\prooforacle{CDH}{}$ more than $q$ times.
  
  Concretely, we build the adversary $\advb$ as follows:
  \begin{description}[topsep=4pt]
    \item[$x \leftarrow \prooforacle{H_{G}}{\advb}(w)$] responds with a new
      challenge element $x \in Ch$ and stores the mapping between each keyword
      and its matching group element.
    \item[$TC \leftarrow \prooforacle{Pub}{\advb}()$] since $\prooforacle{CDH}{}$
      has its own secret $\alpha$, oracle $\prooforacle{Pub}{\advb}$ does not
      choose another secret. The oracle creates $TC$ in the same manner as
      $\prooforacle{Pub}{\sim}$.
    \item[$x^{\alpha} \leftarrow \prooforacle{exp}{\advb}(x)$] uses
      $\prooforacle{CDH}{}(x)$ to respond to up to $q$ queries.
    \item[$\tau \leftarrow \prooforacle{H_{kw}}{\advb}(\omega)$] is similar to
      $\prooforacle{H_{kw}}{\sim}$ and responds to repeated queries
      consistently.  Unlike $\prooforacle{H_{kw}}{\sim}$, this oracle does not
      know the secret $\alpha$ to decrypt the input element. Instead, it uses
      $\prooforacle{DL_\alpha}{}$ to check $z$ against all challenge points $x
      \in Ch$ and find the corresponding element $ z = x^\alpha$ where \(1 =
      \prooforacle{DL_\alpha}{}(x, z) \).
      \begin{enumerate}
        \item Parse the input $\omega$ as ``$d || z$'' where $d \in \Zn$ and $z \in
        \G$. If this fails, respond with a random $l$-bit tag $\tau$.
        \item Find challenge point $x \in Ch$ where $1 =
          \prooforacle{DL_\alpha}{}(x, z)$. If no such point exists, respond
          with a random $l$-bit tag $\tau$.
        \item If $x$ has been queried from the oracle
          $\prooforacle{H_G}{\advb}$, let $w$ be the corresponding preimage.
          Otherwise, respond with a random $l$-bit tag $\tau$.
        \item If $w$ has not been queried, query $w$ from the ideal oracle
          $\prooforacle{\ideal}{}(\mathcal{Y}, \cdot)$ and store the response.
        \item If $\advb$ has queried $q+1$ distinct keywords from the ideal
          oracle, then abort the simulation and solve the One-more-Gap-DH
          challenge.
        \item Respond with a random unused tag $\tau \in TC$ if $d \in
          \prooforacle{\ideal}{}(\mathcal{Y}, w)$. Otherwise, respond with a
          random $l$-bit tag $\tau$.
      \end{enumerate} 
  \end{description}

  In Lemma~\ref{lm:ind-sim-b} (below), we prove that the adversary $\adv$ cannot
  distinguish the simulator $\sim$ from adversary $\advb$. Therefore, if $\adv$
  has a non-negligible chance $\epsilon$ in forcing an abort in $\sim$, then
  with the probability $\epsilon$ adversary $\adv$ queries $q+1$ distinct
  keywords from $\prooforacle{H_{kw}}{\advb}$ while making at most $q$ queries
  from  oracle $\prooforacle{exp}{\advb}$.
  Let \(\{\omega_i\}_{i \in \NatNumUpTo{q+1}}\) be the
  $\prooforacle{H_{kw}}{\advb}$ queries corresponding to the distinct keywords
  \(\{w_i\}_{i \in \NatNumUpTo{q+1}}\). By the construction of
  $\prooforacle{H_{kw}}{\advb}$, we know that \(\displaystyle \left\{\omega_i =
  \text{``}d_i \concat z_i\text{''} \,\middle|\, 1 =
  \prooforacle{DL_\alpha}{}(x_i, z_i) \land x_i =
  \prooforacle{H_G}{\advb}(w_i)\right\}_{i \in \NatNumUpTo{q+1}}\). Since the
  oracle $\prooforacle{H_G}{\advb}$ responds with fresh challenge points, we
  know that the $x_i$s are unique and belong to the challenge set $Ch$.
  Adversary $\advb$ queries the $\prooforacle{CDH}{}$ oracle once per
  $\prooforacle{exp}{\advb}$ query. Since adversary $\adv$ makes less than
  $q+1$ queries from $\prooforacle{exp}{\advb}$, $\advb$ makes at most $q$
  queries from the CDH oracle $\prooforacle{CDH}{}$.
  Adversary $\advb$ produces $q+1$ DH pairs \(\{(x_i, z_i) \mid z_i = x_i^\alpha
  \land x_i \in Ch\}_{i \in \NatNumUpTo{q+1}}\) with at most $q$ queries to the
  CDH oracle $\prooforacle{CDH}{}$ and solves the One-more-Gap challenge with
  probability $\epsilon$.
\end{proof}

\begin{lemma}
  \label{lm:ind-sim-b}
  The adversary $\advb$ is indistinguishable from simulator $\sim$.
\end{lemma}

\begin{proof}
  Oracles $\prooforacle{exp}{\advb}$ and $\prooforacle{Pub}{\advb}$ are
  identical to their $\sim$  counterparts $\prooforacle{exp}{\sim}$ and
  $\prooforacle{Pub}{\sim}$. The oracle $\prooforacle{H_{G}}{\advb}$ responds
  with challenge points which are indistinguishable from the uniformly random
  elements used in $\prooforacle{H_{G}}{\sim}$. Oracles
  $\prooforacle{H_{kw}}{\advb}$ and $\prooforacle{H_{kw}}{\sim}$ only differ in
  how they compute $x = z^{\alpha^{-1}}$ in the step 2. 
  We split the inputs to the oracle $\prooforacle{H_{kw}}{\advb}$ based on their
  inclusion in the challenge set $ch$, and show that oracle
  $\prooforacle{H_{kw}}{\advb}$ is indistinguishable from oracle
  $\prooforacle{H_{kw}}{\sim}$ in both cases. As long as the element $x$ is
  chosen from the challenge set, i.e., $x \in Ch$, the pair $(x, z)$ is unique,
  and both oracles compute the same effective input $x$ because $x =
  z^{\alpha^{-1}}$ is equivalent to $1 = \prooforacle{DL_\alpha}{}(x, z)$. 
  On the other hand, when the element $z$ is generated from an element $x =
  z^{\alpha^{-1}}$ which is not in the challenge set, then oracle
  $\prooforacle{H_{kw}}{\advb}$ cannot compute $x$. Despite the fact that the
  oracle cannot compute $x$, it can determine that $x$ is not from the challenge
  set and consequently not a response from oracle $\prooforacle{H_{G}}{\advb}$
  as $x \notin \text{Range}\left(\prooforacle{H_{G}}{\advb}\right) = Ch$. Both
  oracles $\prooforacle{H_{kw}}{\advb}$ and $\prooforacle{H_{kw}}{\sim}$ respond
  with a random $l$-bit tag when the effective input $x$ is not a response from
  oracles $\prooforacle{H_{G}}{\advb}$ and $\prooforacle{H_{G}}{\sim}$,
  respectively. We conclude that oracles $\prooforacle{H_{kw}}{\advb}$ and
  $\prooforacle{H_{kw}}{\sim}$ are indistinguishable.
\end{proof}
  

%% file: parts/s10-appendix.tex

\newcommand{\docj}{d_j}
\newcommand{\docx}{d_x}
\newcommand{\qi}{q_i}
\newcommand{\qA}{qa}
\newcommand{\qB}{qb}
\newcommand{\qAi}{\qA_i}
\newcommand{\qBi}{\qB_i}
\newcommand{\qAx}[1]{\qA_{#1}}
\newcommand{\qBx}[1]{\qB_{#1}}
\newcommand{\IM}{S}
\newcommand{\im}[2]{\IM_{#1, #2}}
\newcommand{\imij}{\im{i}{j}}
\newcommand{\Bin}{\textvar{Bin}}
\newcommand{\qnum}{m}
\newcommand{\tagsym}{\textvar{tag}}
\newcommand{\ktg}[1]{t_{#1}}
\newcommand{\Tags}{\textvar{Tags}}
\newcommand{\qtags}[1]{{q_{#1}}.\textvar{Tags}}
\newcommand{\qtagsi}{\qtags{i}}
\newcommand{\bitand}{\;\&\;}

\newcommand{\docs}{\textvar{Docs}}
\newcommand{\allkws}{U}
\newcommand{\docrepr}{R}
\newcommand{\uqlim}{\textvar{ulim}}
\newcommand{\oracle}{\mathcal{O}}
\newcommand{\oracleQuery}{\oracle.\textvar{query}}


\section{The limits of document search} 
\label{appendix:sec:theoretical_search}
We show that even with ideal searches an adversary can recover documents or even
extract the whole corpus. We formalize the extraction problem as follows: an
adversary receives a list of $\n$ keywords $\allkws = \kwlist$ and a search
oracle $\oracle$ which respond to queries using the server's set of
$\serversize$ documents $\docs = \{\doc_1, .., \doc_{\serversize}\}$. The
adversary's goal is recovering the document set $\docs$. Since the adversary is
only interested in the set $\allkws$ of keywords, we ignore any keyword outside
of this set in our analysis.

\subsection{One-bit search extraction} 
\label{appendix:sub:one_bit_search}
In this section, we consider a 1-bit search oracle $\oracle$ which returns a
boolean answer for each query which determine whether at least one matching
document exists. The oracle supports one operation, \textsf{query}, which takes
a set of keywords $\partialset$ as input and returns boolean answer $\binanswer
\leftarrow \oracleQuery(\partialset)$. 

A set of keywords $\docrepr$ represent a document if and only if this set
returns a positive search result $\oracleQuery(\docrepr) = 1$ and adding any
other keyword to this set $\docrepr$ results in a negative response $\forall
x \in \allkws, x \notin \docrepr: \oracleQuery(\docrepr \cup \{x\}) = 0$.
It is easy to see that a document $D = \{d_1, .., d_m\}$ is represented by
$\docrepr = D \cap \allkws$.

Recall from Section~\ref{para:search-privacy} that uniqueness number
$\uniqueNum_{D}$ is the smallest number of keywords that uniquely identify a
document $D$. Moreover, when a document $D_x$ is included in a larger document
$D_y$, i.e. $D_x \subset D_y$, then its uniqueness number $\uniqueNum_{D_x}$ is
$\infty$, and document $D_x$ cannot be detected. We have discussed that such
documents do not have a high impact as they are overshadowed by the larger
document. In this section, we assume that all documents have a finite uniqueness
number and only recover documents with uniqueness number $\uniqueNum_D$ less
than the uniqueness limit $\uqlim$.

\para{Document recovery} 
We assume an adversary who has partial knowledge $\partialset$ about a document
$D$ that is represented by $m$ keywords. If the adversary wants to recover the
rest of document $D$, then she needs to ask at least $t =\n-m$ and at most $\n$
queries from the oracle which leads to a $\Theta(n)$ query complexity. It is
important to note that if there is more than one document that contains
$\partialset$, then recovering \emph{any} of these documents counts as document
recovery.

We claim that the adversary needs to ask at least one query for each keyword
which is not in the document, i.e., that it must make at least $t = \n - m$
queries. We assume to the contrary that the adversary recovers the document with
less than $t$ queries and then show that there are two possibilities for $D$
that the adversary cannot distinguish. Since the number of queries is smaller
than $t$, based on the pigeonhole principle a keyword $x$ exist which has never
been queried without another keyword $y \notin D$ present in the query. We claim
that the adversary cannot distinguish $D$ from the document $D \cup \{x\}$ as
the oracle's responses to all queries will be consistent for both documents. The
queries which do not include $x$ are not impacted by the inclusion of $x$ in the
document, and queries that include $x$ include a keyword $y \notin D$ which
ensures a negative answer for both $D$ and $D \cup \{x\}$. Hence, the adversary
cannot distinguish $D$ from $D \cup \{x\}$ and needs to make at least $t$
queries. Clearly, $\n$ queries suffice; showing the result.

\newcommand{\start}{k}
\newcommand{\docset}{D}
\newcommand{\uqset}{P} 
\newcommand{\uqnum}{\uniqueNum}

Algorithm~\ref{alg:RecoverDocument} recovers a document with $n$ queries.
Without loss of generality, we re-index the keywords to represent the
adversary's known set of keywords as $\partialset = \{\kw{1}, ..,
\kw{\start-1}\}$. The algorithm extends this set with the remaining keywords
$\{\kw{\start}, .., \kw{\n} \}$ as long as the oracle keeps returning 1.
Eventually, the algorithm returns a maximal extension of the initial set
$\partialset$.

\begin{algorithm}[tbp]
  \footnotesize
  \caption{Recover the rest of the document given a keyword set $\{\kw{\start}, .., \kw{\n} \}$.\\
  Start : $\textsc{RecoverDocument}(\partialset)$}
  \label{alg:RecoverDocument}
  \begin{algorithmic}
    \Function{RecoverDocument}{$\partialset$}
    \For {$i \leftarrow \start \ldots \n$}
      \If {$\oracleQuery(\partialset \cup \{\kwi\}) = 1$}
          \State $\partialset \leftarrow \partialset \cup \{\kwi\}$ 
      \EndIf
      \EndFor
    \State \Return \partialset
    \EndFunction
  \end{algorithmic}
\end{algorithm}

\para{Corpus extraction} 
\label{appendix:par:one_bit_search_extraction}
Having a set $\partialset$, extracting \emph{one} plausible document is
straightforward. However, extracting \emph{all} documents that contain
$\partialset$ is more complex. The reason behind this complexity is that when
the adversary adds a keyword $\kwx$ to the set $\partialset$ and receives a
positive query response, she knows a document $D$ exists such that
$\left(\partialset \cup \{\kwx\}\right) \subseteq D$ but cannot determine
whether any document $D'$ exists such that $\partialset \subset D'$ and $\kwx
\notin D'$. Hence, the adversary needs to expand both cases.

We designed a corpus extraction algorithm that takes care of this uncertainty,
see Algorithm~\ref{alg:extract-one-bit}. This recursive algorithm is called with
a set of sets representing the documents $\docset$, the set of keywords
$\partialset$ that the algorithm is considering at this moment, and the index
$\start$ into the list of keywords (the keywords with index less than $\start$
have already been considered). To find all documents with respect to the list of
keywords $\kwset=\{a_1, \ldots, a_{\n}\}$, call $\textsc{Extract}(\emptyset,
\emptyset, 1)$.

The algorithm is recursive. It considers the current set of keywords
$\partialset$ and tries to extend it with a keyword $a_i$ ($\start \leq i <
\n$). If the oracle returns 0, clearly there is no document matching
$\partialset \cup \{a_i\}$. If the oracle returns 1, we cannot distinguish the
two cases above, so we recurse along both paths, one for documents that contain
$a_i$, and the other for documents that do not contain $a_i$. 
When the algorithm finds the $\uqlim$'th keyword in the set $\partialset$, the
algorithm can uniquely identify the document and checks whether this document
has been extracted before (by calling \textsf{IsInDocs}) to prevent duplicates.
After reaching a partial set of at least $\uqlim$ keywords, the algorithm only
traverses the branch which includes $a_i$ as only one document exists which
contains the set $\partialset$ since $|\partialset| \geq \uqlim$. When pursuing
only one branch, the algorithm is similar to the \textsf{RecoverDocument}
function in Algorithm~\ref{alg:RecoverDocument}. 
If the algorithm exhausts all possible keywords without branching, it has found
a document and after checking for duplicates, the algorithm adds $\partialset$
as a new document to the current set of documents $\docset$ and returns.

\input{resource/appendix-extraction-alg.tex}

We argue that this algorithm finds all documents with uniqueness number
$\uqnum_D < \uqlim$. Clearly, the algorithm explores all sets $\partialset$ of
size less than $\uqlim$ for which there exist matching documents. So,
eventually, the algorithm will find the unique set for each document, which it
will then extend to the corresponding full document.

It is easy to see that the brute-force part, when $|\partialset| < \uqlim$,
requires at most $\order(\n^{\uqlim})$ queries. However, the algorithm does not
expand keyword sets with negative responses, and on average, document sparsity
leads to a significantly lower number of queries. Once $|\partialset| \geq
\uqlim$ the algorithm enters a linear exploration, as it stops branching. It
runs through this linear phase exactly once for each document. Resulting in a
total complexity of $\order(\n^{\uqlim} + \n{}d)$.

\subsection{\numdocsearch search extraction} 
\label{appendix:sub:numdocsearch_search_extraction}
In this section, we consider a \numdocsearch search oracle $\oracle$ which
returns the number of matching documents for each query. The oracle only
supports one operation, \textsf{query}, which takes a set of keywords
$\partialset$ as input and returns the number of matching documents $t
\leftarrow \oracleQuery(\partialset)$.

\para{Document recover}
Since the 1-bit search oracle's output can be computed from the \numdocsearch
oracle, the algorithms from the previous section also work against the
\numdocsearch search oracle. As a matter of fact, when only considering a single
document, the behavior of the \numdocsearch oracle is equivalent to that of the
1-bit search oracle, thus \textsc{RecoverDocument} in
Algorithm~\ref{alg:RecoverDocument} is also optimal for the \numdocsearch system
in recovering documents.

\para{Corpus extraction} 
\label{appendix:par:num_doc_search_extraction}
The extra information provided by the \numdocsearch oracle, however, helps
create a much more efficient corpus extraction function. In particular, an
attacker is no longer faced with the uncertainty caused by the one-bit oracle.
Given an existing set of keywords $\partialset$, the attacker can query
$\partialset \cup \{\kwx\}$ and see if the number of matching documents changes,
or not. If the number of matching documents changes, there were documents that
match $\partialset$ but not $\partialset \cup \{\kwx\}$. If the number of
matching documents stays the same, all documents that match $\partialset$ also
match $\partialset \cup \{\kwx\}$.

\input{resource/appendix-extraction-docnum-alg.tex}

Algorithm~\ref{alg:extract-docnum} exploits this principle. It keeps track of
the current set of documents represented as $\docset$, the set of keywords
$\partialset$ that it is currently considering, the index $\start$ into the list
of keywords (the keywords with index less than $\start$ have already been
considered), and the number \textsf{matches} of documents that contain the
current set of keywords $\partialset$. To find all documents with respect to the
set of keywords $\kwset=\{a_1, \ldots, a_{\n}\}$, call
$\textsc{Extract}(\emptyset, \emptyset, 1, \infty)$.

Given the current set $\partialset$ with \textsf{matches} matching documents it
proceeds as follows. It asks the next keyword $a_i$, if there are still matching
documents (i.e., $\textsf{next} > 0$ ) it adds $a_i$ to $\partialset$ and
continues exploring. If some documents matched $\partialset$ but did not match
$\partialset \cup a_i$ (i.e., $\textsf{matches} > \textsf{next}$), the algorithm
also continues exploring by skipping the keyword $a_i$.

In the beginning, the algorithm starts with an empty set and checks every
keyword. This requires $\n$ queries, and the algorithm continues with a
deterministic document recovery for $d$ documents. Therefore, this algorithm
requires a total of $\order(\n{}d)$ queries for extracting the corpus.

%% file: resource/appendix-extraction-alg.tex
\begin{algorithm}[tbp]
	\footnotesize
	\caption{Extract non-contained documents with an uniquness number $\uniqueNum_D$
	smaller than $\uqlim$ with a one-bit search oracle based on the keyword set
	$\kwset=\{\kw{1}, .., \kw{\n}\}$.\\
	Start: $\textsc{Extract}(\emptyset, \emptyset, 1)$}
	\label{alg:extract-one-bit}
	\begin{algorithmic}
		\Function{Extract}{$\docset, \partialset,\start$}
		\If {$|\partialset| \geq \uqlim$}
			\Comment The document is uniquely identifiable.
			\If {$\textsc{IsInDocs}(\partialset, \docset) = 1$}
				\Comment $\partialset$ is already extracted.
				\State \Return \docset
			\EndIf
		\EndIf

		\For {$i \leftarrow \start \ldots \n$}
			\If {$\oracleQuery(\partialset \cup \{\kwi\}) = 1$}
				\State $\docset \leftarrow \textsc{Extract}(\docset, \partialset \cup \{\kwi\}, i+1)$ 
				\If {$|\partialset| < \uqlim$}
					\State $\docset \leftarrow \textsc{Extract}(\docset, \partialset, i+1)$ 
				\EndIf
				\State \Return  $\docset$
			\EndIf
		\EndFor

		\If {$\textsc{IsInDocs}(\partialset, \docset) = 0$}
		\Comment No more extension possible.
		\State $\docset \leftarrow \docset \cup \{\partialset\}$
		\EndIf
		\State \Return $\docset$
		\EndFunction
		\\
		\Function{IsInDocs}{$\partialset,\docset$}
		\ForAll {$\doc \in \docset$}
			\If {$\partialset \subseteq \doc$}
				\State \Return 1
			\EndIf
			\EndFor
		\State \Return 0
		\EndFunction
	\end{algorithmic}
\end{algorithm}

%% file: resource/appendix-extraction-docnum-alg.tex
\newcommand{\matches}{\textvar{matches}} 
\newcommand{\newmatches}{\textvar{next}}

\begin{algorithm}[tbp]
	\footnotesize
	\caption{Extract all $\#\matches$ documents which include the partial document
	$\partialset$, with a \numdocsearch search oracle based on the keyword set
	$\kwset=\{\kw{1}, .., \kw{\n}\}$.\\
	Start : $\textsc{Extract}(\emptyset, \emptyset, 1, \infty)$}
	\label{alg:extract-docnum}

	\begin{algorithmic}
		\Function{Extract}{$\docset, \partialset, \start, \matches$}
		\For {$i \leftarrow \start,\n$}
			\State $\newmatches = \oracleQuery(\partialset \cup \{\kwi\})$
			\If {$ \newmatches > 0$}
				\State $\docset \leftarrow \textsc{Extract}(\docset, \partialset \cup \{\kwi\}, i+1, \newmatches)$ 
				\If {$\matches > \newmatches$}
					\Comment At least one doc did not contain $\kwi$
					\State $\docset \leftarrow \textsc{Extract}(\docset, \partialset, i+1, \matches-\newmatches)$ 
				\EndIf
				\State \Return $\docset$ \vspace{0.4em}
			\EndIf 
		\EndFor 
		\Return $\docset \cup \{\partialset\}$
		\EndFunction
	\end{algorithmic}

\end{algorithm}